\newif\iflongversion
\longversiontrue

\iflongversion
\documentclass[conference]{IEEEtran}
\else
\documentclass[conference]{IEEEtran}
\fi

\IEEEoverridecommandlockouts

\usepackage{cite}
\usepackage{algorithm}
\usepackage{graphicx}
\usepackage{textcomp}
\usepackage{xcolor}
\def\BibTeX{{\rm B\kern-.05em{\sc i\kern-.025em b}\kern-.08em
    T\kern-.1667em\lower.7ex\hbox{E}\kern-.125emX}}

\usepackage{amsmath, amssymb, amsthm}
\usepackage{nicematrix}
\usepackage{esint}
\usepackage{braket}
\usepackage{mathrsfs}
\usepackage{float}
\usepackage{enumitem}

\usepackage{tikz}
\usetikzlibrary{positioning, calc, optics, arrows, chains, matrix, scopes, shapes, external}
\pgfdeclarelayer{background}
\pgfdeclarelayer{behind}
\pgfdeclarelayer{above}
\pgfdeclarelayer{glass}
\pgfsetlayers{background, behind, main, above, glass}
\tikzset{
fnode/.style={rectangle, draw, minimum width=12mm, minimum height=8mm, anchor=center, align=center, node distance=7mm}, font=\sffamily
}

\usepackage[modulo, switch]{lineno}

\usepackage{indentfirst}
\usepackage{algpseudocode}
\usepackage{makecell}

\theoremstyle{plain}
\newtheorem{thm}{Theorem}
\newtheorem{lem}[thm]{Lemma}
\newtheorem{prop}[thm]{Proposition}
\newtheorem{corollary}[thm]{Corollary}

\theoremstyle{definition}

\newtheorem{example}[thm]{Example}
\newtheorem{definition}[thm]{Definition}
\newtheorem{assumption}[thm]{Assumption}

\newtheorem{remark}[thm]{Remark}
\newcommand{\edefn}{\hfill$\triangle$}

\providecommand{\N}{\mathbb {N}}
\providecommand{\Z}{\mathbb {Z}}
\providecommand{\R}{\mathbb {R}}
\providecommand{\C}{\mathbb {C}}

\newcommand{\T}{\mathrm{T}}
\newcommand{\B}{\mathrm{B}}
\newcommand{\SPA}{\mathrm{SPA}}
\newcommand{\E}{\mathbb {E}}
\newcommand{\tr}{\operatorname{tr}}
\newcommand{\stat}{\operatornamewithlimits{stat}}

\newcommand{\RomanNumeralCaps}[1]{\MakeUppercase{\romannumeral #1}}

\newcommand{\sfn}{\mathsf{N}}
\newcommand{\fr}[1]{ f^{\mathrm{R}}_{ #1 } }
\newcommand{\fl}[1]{ f^{\mathrm{L}}_{ #1 } }
\newcommand{\fe}[1]{ f_{ #1 }}
\newcommand{\xr}[1]{ x^{\mathrm{R}}_{ #1 } }
\newcommand{\xl}[1]{ x^{\mathrm{L}}_{ #1 } }

\newcommand{\xlprime}[1]{ x^{\mathrm{L}'}_{ #1 } }

\newcommand{\tfr}[1]{\widetilde{f}^{\mathrm{R}}_{ #1 }}
\newcommand{\tfl}[1]{\widetilde{f}^{\mathrm{L}}_{ #1 }}
\newcommand{\tfe}[1]{\widetilde {f}_{ #1 }}
\newcommand{\txr}[1]{\widetilde{x}^{\mathrm{R}}_{ #1 }}
\newcommand{\txl}[1]{\widetilde{x}^{\mathrm{L}}_{ #1 }}

\newcommand{\hatfr}[1]{\widehat{f}^{\mathrm{R}}_{ #1 }}
\newcommand{\hatfl}[1]{\widehat{f}^{\mathrm{L}}_{ #1 }}
\newcommand{\hatx}[1]{\widehat{x}_{ #1 }}
\newcommand{\hatbfx}[1][]{\widehat{\mathbf{x}}^{ #1 }}

\newcommand{\bbeta}{\boldsymbol{\beta}}
\newcommand{\br}[1]{\bbeta^{\mathrm{R}}_{#1}}
\newcommand{\bl}[1]{\bbeta^{\mathrm{L}}_{#1}}
\newcommand{\be}[1]{\bbeta_{#1}}
\newcommand{\muleft}[2]{\overleftarrow{\boldsymbol{\mu}}^{\mathrm{#1}}_{#2}}
\newcommand{\muright}[2]{\overrightarrow{\boldsymbol{\mu}}^{\mathrm{#1}}_{#2}}

\newcommand{\perm}{\operatorname{perm}}

\newcommand{\btheta}{\boldsymbol{\theta}}

\newcommand{\defeq}{\triangleq}
\newcommand{\sfnDE}{\sfn_{\mathrm{DE}}}
\newcommand{\bepsilon}{\boldsymbol{\epsilon}}

\newcommand{\validc}{\mathcal {C}(\widehat{\mathcal N}_2)}

\begin{document}


\title{Complex-Valued-Matrix Permanents: SPA-based Approximations
 and Double-Cover Analysis
}


\author{\IEEEauthorblockN{Junda Zhou and Pascal O.~Vontobel}
 \IEEEauthorblockA{\textit{Department of Information Engineering} \\
  \textit{The Chinese University of Hong Kong}\\
  davidjdzhou@link.cuhk.edu.hk, pascal.vontobel@ieee.org}
}


\maketitle


\iflongversion
 \thispagestyle{plain}
 \pagestyle{plain}
\fi


\begin{abstract}

 Approximating the permanent of a complex-valued matrix is a fundamental
 problem with applications in Boson sampling and probabilistic inference. In
 this paper, we extend factor-graph-based methods for approximating the
 permanent of non-negative-real-valued matrices that are based on running the
 sum-product algorithm (SPA) on standard normal factor graphs, to
 factor-graph-based methods for approximating the permanent of complex-valued
 matrices that are based on running the SPA on double-edge normal factor
 graphs.

 On the algorithmic side, we investigate the behavior of the SPA, in particular
 how the SPA fixed points change when transitioning from real-valued to
 complex-valued matrix ensembles. On the analytical side, we use graph covers
 to analyze the Bethe approximation of the permanent, i.e., the approximation
 of the permanent that is obtained with the help of the SPA.

 This combined algorithmic and analytical perspective provides new insight into
 the structure of Bethe approximations in complex-valued problems and clarifies
 when such approximations remain meaningful beyond the non-negative-real-valued
 settings.
\end{abstract}




\section{Introduction}



\subsection{Permanents}


For any $n\in\Z_{\geq1}$, let $[n] \defeq \set{1, \dots, n}$ and let
$\mathcal{S}_n$ be the symmetric group of $[n]$, i.e., the set of all
permutations of $[n]$. For a complex-valued matrix
$\btheta=(\theta_{i,j})_{i,j}$ of size $n \times n$, the permanent of $\btheta$
is defined to be
\[\perm(\btheta) \defeq \sum_{\sigma\in \mathcal{S}_n}\prod_{i\in[n]}
 \theta_{i,\sigma(i)}.\] In this work, we study approximations of $\perm (\btheta)$ when the entries
   of $\btheta$ are i.i.d.~complex-valued random variables.


Note that already the calculation of the permanent of non-negative-real-valued
matrices is computationally intractable for large $n$, and also its efficient
approximation is challenging. (For more details, see, e.g., the discussion in
\cite{vontobelBethePermanentNonnegative2013}.) The approximate computation of
the permanent of complex-valued matrices is even more challenging, as the real
and imaginary parts of $\prod_{i\in[n]} \theta_{i,\sigma(i)}$ can be positive,
zero, or negative, and with that leading to ``constructive and destructive
interferences'' when approximately computing the real and imaginary parts of
$\sum_{\sigma\in \mathcal{S}_n}\prod_{i\in[n]} \theta_{i,\sigma(i)}$.


The permanent of a complex-valued matrix arises, for example in Boson sampling,
a candidate proposed \cite{aaronsonComputationalComplexityLinear2013} for
demonstrating quantum supremacy. In Boson sampling, the measurement outcome of
a linear optics system follows a distribution with an exponential-size support,
and each non-zero probability depends on the permanent of a complex matrix.
Based on conjectures on the complexity and distribution of permanents of random
complex matrices, it is proven \cite{aaronsonComputationalComplexityLinear2013}
that no classical algorithm can efficiently sample (exactly or approximately)
from the outcome distribution, unless certain complexity classes collapse.


Towards resolving the complexity conjectures, progress
\cite{aaronsonBosonSamplingFarUniform2013,eldarApproximatingPermanentRandom2018,
jiApproximatingPermanentRandom2021,nezamiPermanentRandomMatrices2021,chabaudResourcesBosonicQuantum2023}
has been made to identify the sources of computational difficulties and push
the limit of classical methods. In this light, we apply the techniques
developed for normal factor graphs
(NFGs)~\cite{forneyCodesGraphsNormal2001,loeligerIntroductionFactorGraphs2004}
for non-negative-real-valued matrices to complex-valued matrices.


\subsection{Factor graphs and the sum-product algorithm}


A factor graph
\cite{kschischangFactorGraphsSumproduct2001,loeligerIntroductionFactorGraphs2004}
represents the factorization structure of a multivariate function. More
precisely, for a so-called global function that can be written as the product
of so-called local functions, a factor graph shows the local functions and
their variable dependencies by drawing a function node for every local
function, a variable node for every variable, and an edge if a variable appears
as an argument of a local function. The partition sum (or partition function)
of a factor graph is the sum of the global function over all configurations,
i.e., all possible assignments of values to the variables.


For a factor graph whose partition sum is computationally intractable, various
techniques
\cite{wainwrightGraphicalModelsExponential2008,yedidiaConstructingFreeenergyApproximations2005}
were proposed for efficient approximations of the partition sum. In particular,
the so-called Bethe approximation is obtained by running the sum-product
algorithm (SPA), also known as loopy belief propagation, until convergence, and
doing certain calculations based on the SPA fixed-point messages. In the
context of a permanent for a non-negative-real-valued matrix $\btheta$, a
normal factor graph $\sfn(\btheta)$ can be formulated such that its partition
sum equals $\perm(\btheta)$ (see, e.g.,
\cite{vontobelBethePermanentNonnegative2013}). The resulting Bethe
approximation is known as the Bethe permanent $\perm_\B(\btheta)$, and has been
analyzed extensively through the lenses of variational analysis, the SPA, and
graph
covers~\cite{vontobelBethePermanentNonnegative2013,vontobelCountingGraphCovers2013,
ngDoublecoverbasedAnalysisBethe2022,huangDegreeMBetheSinkhorn2024,vontobelUnderstandingRatioPartition2025}.


Toward better understanding the Bethe approximation of a partition sum, it is
of particular interest to study the ratio $Z(\sfn) / Z_\B(\sfn)$, where
$Z(\sfn)$ represents the partition sum of a normal factor graph $\sfn$ and
where $Z_\B(\sfn)$ represents the Bethe approximation of $Z(\sfn)$. For various
classes of factor graphs it has been observed
that~\cite{vontobelUnderstandingRatioPartition2025}
\begin{align}
 \frac{Z(\sfn)}{Z_\B(\sfn)}
  & \approx
 \biggl( \frac{Z(\sfn)}{Z_{\B,2}(\sfn)} \biggr)^{\!\! 2},
 \label{eq:Bethe:approximation:ratio:1}
\end{align}
where $Z_{\B,2}(\sfn)$ is the so-called degree-$2$ Bethe approximation of
$Z(\sfn)$~\cite{vontobelCountingGraphCovers2013}. The expression
in~\eqref{eq:Bethe:approximation:ratio:1} is of interest because, while one
would like to characterize the ratio $Z(\sfn) / Z_\B(\sfn)$ on the left-hand
side of~\eqref{eq:Bethe:approximation:ratio:1}, this ratio is usually harder
to characterize than the ratio $Z(\sfn) / Z_{\B,2}(\sfn)$ appearing on the
right-hand side of~\eqref{eq:Bethe:approximation:ratio:1}.


\subsection{Double-edge normal factor graphs}


Double-edge normal factor graphs (DE-NFGs) were proposed in
\cite{caoDoubleedgeFactorGraphs2017} with the goal of having factor graphs
whose partition sum represents quantities of interest in quantum information
processing, and in particular toward computing Bethe approximation of the
partition sum with the help of the SPA. The latter was analyzed
in~\cite{huangCharacterizingBethePartition2020,
huangGraphcoverbasedCharacterizationBethe2025} for general setups.


Let $\btheta$ be some complex-valued matrix. In the present work, we consider
the DE-NFG $\sfnDE(\btheta)$ that is designed such that its partition sum
equals $|\perm(\btheta)|^2$. We study the SPA for $\sfnDE(\btheta)$ and analyze
the resulting Bethe approximations of $|\perm(\btheta)|^2$. In particular, we
want to investigate whether the expression
in~\eqref{eq:Bethe:approximation:ratio:1} also holds for this class of DE-NFGs.


\subsection{Notation and supplementary materials}


We use $\mathbf 0_n $ and $\mathbf 1_n$ to denote the all-zero and all-one
matrix of size $n\times n$, respectively. We use iota ``$\iota$'' for the
imaginary unit. The set of all non-negative integers is denoted as $\N$. For a
statement $P$, we use the Iverson bracket to indicate its truth, i.e.,
$[P]\defeq 1$ if $P$ is true, $[P]\defeq 0$ otherwise.

\iflongversion Supplementary materials and proofs are given in the appendices.
\else Supplementary materials and proofs can be found in the extended version
 \cite{zhou2026}.  \fi


\section{NFG and DE-NFG Representations}
\label{sec:preliminary}


In this paper we will use a variant of factor graphs called normal factor
graphs (NFGs). In NFGs, variables are associated with edges and half-edges.
While at first this might seem limiting, the global function of any factor
graph can be reformulated such that all variable nodes have degree two or one,
and so these variable nodes can effectively be omitted when drawing the factor
graph. Similarly, we will use a variant of double-edge factor graphs called
double-edge normal factor graphs (DE-NFGs).


\begin{figure}[t]
 \begin{minipage}{.23\textwidth}
\centering
\begin{tikzpicture}[on grid, auto]
\tikzstyle{state}=[shape=rectangle, draw, minimum size=4mm]
\pgfmathsetmacro{\hdis}{2.4} 
\pgfmathsetmacro{\vdis}{0.8} 
    \foreach \i in {1,2,3,4,5}{
        \node[state, label=left:$\fl{\i}$] (f\i) at (0,-\i*\vdis) {};
        \node[state, label=right:$\fr{\i}$] (g\i) at (\hdis,-\i*\vdis) {};
    }
    \foreach \i in {1,2,3,4,5}{
        \foreach \j in {1,2,3,4,5}{                
                 \draw[] (f\i) -- (g\j) ;
        }
    }
\end{tikzpicture}
\end{minipage}\hfill
\begin{minipage}{.23\textwidth}
\centering
\begin{tikzpicture}[on grid, auto]
\tikzstyle{state}=[shape=rectangle, draw, minimum size=4mm]
\pgfmathsetmacro{\hdis}{2.4} 
\pgfmathsetmacro{\vdis}{1} 
\node[state, label=left:$\fl{i}$] (f) at (0,-2*\vdis) {};
\node[state, label=right:$\fr{i}$] (g) at (\hdis,0) {};
\node[state, minimum size=3mm] (eq) at (0.5*\hdis, -1*\vdis) {$=$};
\node[state, minimum size=3mm, label=above:$w_{i,j}$] (w) at (0.5*\hdis, -.5*\vdis) {};
\draw[] (f) -- (eq) node [pos=0.5, above, sloped] {$\xl{i,j}$};
\draw[] (eq) -- (g) node [pos=0.5, below, sloped] {$\xr{i,j}$};
\draw[] (w) -- (eq);

\draw[magenta, thick, dashed] ($(w.north west)+(-0.2,0.5)$) rectangle ($(eq.south east)+(0.1,-0.2)$) node [anchor=north] {$\fe{i,j}$};

\foreach \i in {0.33, 0.66}{
    \draw[] (f) -- ++ (0.25*\hdis, -\i*\vdis/2);
    \draw[dashed] ($(f)+(0.25*\hdis, -\i*\vdis/2)$) -- ++ (0.25*\hdis, -\i*\vdis/2);
    \draw[] (g) -- ++ (-0.25*\hdis, \i*\vdis/2);
    \draw[dashed] ($(g)+(-0.25*\hdis, \i*\vdis/2)$) -- ++ (-0.25*\hdis, \i*\vdis/2);
}
\end{tikzpicture}
\end{minipage}
 \caption{NFG for $\perm(\btheta)$ when $n=5$.
  Left: $\sfn(\btheta)$ with edge weights $\fe{i,j}$ omitted.
  Right: Zoomed in on parts of $\sfn(\btheta)$.}\label{fig:NFGperm}
\end{figure}


We start by recalling the definition of the NFGs and DE-NFGs whose partition
sum are related to the matrix permanent
\cite{vontobelBethePermanentNonnegative2013,caoDoubleedgeFactorGraphs2017}. Fix
some positive integer $n$ and let $\btheta$ be a non-negative-real-valued
matrix or complex-valued matrix of size $n\times n$. The NFG $\sfn(\btheta)$
used in this paper (see Fig.~\ref{fig:NFGperm}) is a slightly modified version
of the NFG used in
\cite{vontobelBethePermanentNonnegative2013}:\footnote{In~\cite{vontobelBethePermanentNonnegative2013},
the matrix entries of $\btheta$ appeared in the definition of the ``left check
node'' function and the ``right check node'' functions, whereas here they
appear in the definition of the ``edge weight'' functions.}
\begin{itemize}

 \item The NFG $\sfn({\btheta})$ is based on a complete bipartite graph with
 $2n$ vertices.

 \item For every $(i,j) \in [n]^2$, let $\xl{i,j}$ be the variable associated
 with the edge connecting $\fl{i}$ with $\fe{i,j}$, and let $\xr{i,j}$ be the
 variable associated with the edge connecting $\fr{i}$ with $\fe{i,j}$. All
 variables take values in the set $\mathcal {X} \defeq \set{0,1}$.

 \item For every $i \in [n]$, let the ``left check node'' function be
 \begin{align*}
  \hspace{-1.5em}
  \fl{i}\bigl(\{\xl{i,j}\}_{j\in[n]}\bigr) \defeq
  \begin{cases}
  1 & \text{exactly one of } \{\xl{i,j}\}_{j\in[n]} \text{ equals } 1 \\
  0 & \text{otherwise}\end{cases}.
 \end{align*}
 For every $j \in [n]$, let the ``right check node'' function be
 \begin{align*}
  \hspace{-1.5em}
  \fr{j}\bigl(\{\xr{i,j}\}_{i\in[n]}\bigr) \defeq
  \begin{cases}
  1 & \text{exactly one of } \{\xr{i,j}\}_{i\in[n]} \text{ equals } 1 \\
  0 & \text{otherwise}\end{cases}.
 \end{align*}
 For every $(i,j) \in [n]^2$, let the ``edge weight'' function be
 \begin{align*}
  \fe{i,j}\bigl(\xl{i,j},\xr{i,j}\bigr)\defeq \left[\xl{i,j}=\xr{i,j}\right] \cdot w_{i,j} (\xl{i,j}),
 \end{align*}
 where $w_{i,j}(0) \defeq 1$ and $w_{i,j}(1) \defeq \theta_{i,j}$.

 \item The global function is defined to be
 \begin{align*}
  g & \bigl(\{\xl{i,j},\xr{i,j}\}_{i,j\in[n]^2}\bigr)
  \defeq
  \left(\prod_{i \in [n]}\fl{i}\bigl(\{\xl{i,j}\}_{j\in[n]}\bigr)\right) \\
    &
  \cdot\left(\prod_{i,j \in [n]^2}\fe{i,j}\bigl(\xl{i,j},\xr{i,j}\bigr)\right)
  \cdot\left(\prod_{j \in [n]}\fr{j}\bigl(\{\xr{i,j}\}_{i\in[n]}\bigr)\right).
 \end{align*}
 \item The partition sum (or partition function) is defined to be
 \begin{align*}
  Z(\sfn(\btheta)) & \defeq\sum_{\{\xl{i,j},\xr{i,j}\}_{i,j\in[n]^2}} g(\{\xl{i,j},\xr{i,j}\}_{i,j\in[n]^2}).
 \end{align*}
\end{itemize}
One can verify that $Z\bigl( \sfn(\btheta) \bigr) = \perm(\btheta)$.


Let $\btheta$ be as above and let $\bepsilon$ be a non-negative-real-valued
matrix of size $n\times n$. The DE-NFG $\sfnDE(\btheta,\boldsymbol\epsilon)$
representation \cite{caoDoubleedgeFactorGraphs2017} (see
Fig.~\ref{fig:DENFGperm}) extends $\sfn(\btheta)$ by replacing the binary
variables $\xl{i,j}$ and $\xr{i,j}$ with $\txl{i,j}$ and $\txr{i,j}$,
respectively, which take values in $\mathcal X^2$. We write
$\txl{i,j}=(\xl{i,j},\xlprime{i,j})$, and similarly for $\txr{i,j}$. The
function nodes in $\sfnDE(\btheta,\bepsilon)$ are defined as follows.
\begin{itemize}

 \item For every $i \in [n]$, let the ``left check node'' function be
 \begin{align*}
  \tfl{i}\bigl(\{\txl{i,j}\}_{j\in[n]}\bigr) \defeq \fl{i}\bigl(\{ \xl{i,j} \}_{j \in [n]}\bigr)\cdot \fl{i}\bigl(\{ x^{\mathrm{L}'}_{i,j} \}_{j \in [n]}\bigr).
 \end{align*}

 \item For every $j \in [n]$, let the ``right check node'' function be
 \begin{align*}
  \tfr{j}\bigl(\{\txr{i,j}\}_{i\in[n]}\bigr) \defeq \fr{j}\bigl(\{ \xr{i,j} \}_{i \in [n]}\bigr)\cdot \fr{j} \bigl(\{ x^{\mathrm{R}'}_{i,j} \}_{i \in [n]}\bigr).
 \end{align*}

 \item For every $(i,j) \in [n]^2$, let the ``edge weight'' function be
 \begin{align*}
  \tfe{i,j}\bigl(\txl{i,j},\txr{i,j}\bigr)\defeq  \left[\txl{i,j}=\txr{i,j}\right]\cdot \mathbf W_{i,j}\bigl(\txl{i,j}\bigr),
 \end{align*}
 where
 $\mathbf W_{i,j}
  \defeq
  \Bigl[
   \begin{smallmatrix}
    1            & \overline{\theta_{i,j}} \\
    \theta_{i,j} & |\theta_{i,j}|^2+\epsilon_{i,j}
   \end{smallmatrix}
   \Bigr]$,
 whose row and column entries are indexed by $\xl{i,j}$ and
 $\xlprime{i,j}$, respectively.
\end{itemize}


\begin{figure}[t]
 \begin{minipage}{.23\textwidth}
\centering
\begin{tikzpicture}[on grid, auto]
\tikzstyle{state}=[shape=rectangle, draw, minimum size=4mm]
\pgfmathsetmacro{\hdis}{2.4} 
\pgfmathsetmacro{\vdis}{0.8} 
    \foreach \i in {1,2,3,4,5}{
        \node[state, label=left:$\tfl{\i}$] (f\i) at (0,-\i*\vdis) {};
        \node[state, label=right:$\tfr{\i}$] (g\i) at (\hdis,-\i*\vdis) {};
    }
    \foreach \i in {1,2,3,4,5}{
        \foreach \j in {1,2,3,4,5}{                
                 \draw[double] (f\i) -- (g\j) ;
        }
    }
\end{tikzpicture}
\end{minipage}\hfill
\begin{minipage}{.23\textwidth}
\centering
\begin{tikzpicture}[on grid, auto]
\tikzstyle{state}=[shape=rectangle, draw, minimum size=4mm]
\pgfmathsetmacro{\hdis}{2.4} 
\pgfmathsetmacro{\vdis}{1} 
\node[state, label=left:$\tfl{i}$] (f) at (0,-2*\vdis) {};
\node[state, label=right:$\tfr{j}$] (g) at (\hdis,0) {};
\node[state, minimum size=3mm] (eq) at (0.5*\hdis, -1*\vdis) {$=$};
\node[state, minimum size=3mm, label=above:$\mathbf W_{i,j}$] (w) at (0.5*\hdis, -.5*\vdis) {};
\draw[double] (f) -- (eq) node [pos=0.5, above, sloped] {$\txl{i,j}$};
\draw[double] (eq) -- (g) node [pos=0.5, below, sloped] {$\txr{i,j}$};
\draw[double] (w) -- (eq);

\draw[magenta, thick, dashed] ($(w.north west)+(-0.3,0.55)$) rectangle ($(eq.south east)+(0.2,-0.2)$) node [anchor=north] {$\tfe{i,j}$};

\foreach \i in {0.33, 0.66}{
    \draw[double] (f) -- ++ (0.25*\hdis, -\i*\vdis/2);
    \draw[double, dashed] ($(f)+(0.25*\hdis, -\i*\vdis/2)$) -- ++ (0.25*\hdis, -\i*\vdis/2);
    \draw[double] (g) -- ++ (-0.25*\hdis, \i*\vdis/2);
    \draw[double, dashed] ($(g)+(-0.25*\hdis, \i*\vdis/2)$) -- ++ (-0.25*\hdis, \i*\vdis/2);
}
\end{tikzpicture}
\end{minipage}
 \caption{DE-NFG for $|\perm(\btheta)|^2$ when $n=5$.
  Left: $\sfnDE(\btheta)$ with edge weights $\tfe{i,j}$ omitted.
  Right: Zoomed in on parts of $\sfnDE(\btheta)$.}\label{fig:DENFGperm}
\end{figure}


The global function and partition sum for $\sfnDE(\btheta, \bepsilon)$ are
defined in the same way as for $\sfn(\btheta)$. Notice that when
$\bepsilon=\mathbf{0}_{n}$, the matrix $\mathbf W_{i,j}$ has rank one, and we
then denote $\sfnDE(\btheta, \bepsilon)$ by $\sfnDE(\btheta)$. One can verify
that
\begin{align*}
 Z\bigl( \sfnDE(\btheta) \bigr)
  & = \bigl| \perm(\btheta) \bigr|^2,
 \quad
 Z\bigl(\sfnDE(\mathbf{0}_{n},\bepsilon)\bigr)
 = \perm(\bepsilon).
\end{align*}


\section{Sum-Product Algorithm on DE-NFG}\label{sec:SPA}


Finding the log partition sum for a general NFG (with function nodes taking on
non-negative real values) is equivalent to minimizing the Gibbs free energy
function
\cite{yedidiaConstructingFreeenergyApproximations2005,wainwrightGraphicalModelsExponential2008}.
A computationally tractable relaxation is to replace the Gibbs free energy
function with the Bethe free energy function and enlarge the feasible region.
One efficient algorithm to minimize the Bethe free energy function is the
sum-product algorithm (SPA) \cite{kschischangFactorGraphsSumproduct2001}. This
line of method has been analyzed in detail for the specific NFG $\sfn(\btheta)$
in \cite{vontobelBethePermanentNonnegative2013} and shown to work very well.


Most of the relevant concepts involved in the SPA extend naturally from NFGs to
DE-NFGs (see \iflongversion Appendices \ref{appx:Bethe:details} \&
\ref{appx:spa:details} \else \cite{zhou2026} \fi for $\sfnDE(\btheta)$ and
\cite{huangGraphcoverbasedCharacterizationBethe2025} for general DE-NFGs).
However, generalizing the (primal) Bethe free energy function from NFGs to
DE-NFGs is challenging, as the latter involves complex-valued functions. In
\iflongversion Appendix \ref{appx:Bethe:details}, \else \cite{zhou2026}, \fi we
propose a generalization of the Bethe free energy function for
$\sfnDE(\btheta)$. Although its stationarity conditions can be obtained with
the help of Wirtinger calculus (see e.g.,
\cite{kreutz-delgadoComplexGradientOperator2009}), it does not have the desired
convexity properties. This is similar to the well-known stationary-action
principle from physics (sometimes imprecisely called the minimum-action
principle).


Various details of running the SPA on $\sfnDE(\btheta)$ and
$\sfnDE(\btheta,\boldsymbol\epsilon)$, such as message update rules, are
summarized in \iflongversion Appendix~\ref{appx:spa:details}. \else
\cite{zhou2026}. \fi We use a flooding schedule for message updates, and for
obtaining convergence of the SPA, a mid-point damping rule is applied, i.e.,
\[
 \boldsymbol{\mu}^\mathrm{outgoing}_{t}=\frac{1}{2}\bigl(f_{\mathrm{message\
    update}}(\set{\boldsymbol{\mu}^\mathrm{incoming}_{t-1}})+\boldsymbol{\mu}^\mathrm{outgoing}_{t-1}\bigr)
 .\] Since the edge weights $\mathbf W_{i,j}$ are rank-$1$ in $\sfnDE(\btheta)$,
   the DE-NFG $\sfnDE(\btheta)$ is essentially the disjoint union of two NFGs,
   $\sfn(\btheta)$ and $\sfn(\overline{\btheta})$. To counteract, the SPA
   messages are initialized to be randomly generated rank-$2$ positive
   semi-definite (PSD) matrices (see \iflongversion Appendix
   \ref{appx:spa:details}). \else \cite{zhou2026}). \fi


We study the behavior of the SPA for different ensembles of matrices $\btheta$,
where the entries of $\btheta$ are i.i.d.\ according to some distribution (see
\iflongversion Appendix \ref{appx:random:matrix} \else \cite{zhou2026} \fi for
details). Importantly, the distributions are parameterized by an angle $\alpha
\in [0,\pi]$ and have support on the unbounded sector
\begin{align*}
 S_\alpha
  & \defeq
 \bigl\{
 r e^{\iota \phi}
 \bigm|
 r \ge 0,\ \phi \in [-\alpha,\alpha]
 \bigr\}.
\end{align*}
Fig.~\ref{fig:DENFG:SPA:alpha} shows the numerical results for
$Z \defeq Z(\sfnDE(\btheta))$ and the SPA fixed-point-based Bethe
approximation $Z_{\B,\SPA} \defeq Z_{\B,\SPA}(\sfnDE(\btheta))$. We choose
$21$ values of $\alpha$, equally spaced in $[0,\pi]$, and for each value of
$\alpha$, generate $1000$ matrices $\btheta$ of size $n\times n$.  In
Fig.~\ref{fig:DENFG:SPA:alpha} we plot the empirical values for
\begin{align*}
 \ln
 \biggl( \frac{\E[Z]}{\E[Z_{\B,\SPA}]} \biggr),
 \quad
 \ln
 \biggl(
 \E\biggl[\frac{Z}{Z_{\B,\SPA}}\biggr]
 \biggr),
 \quad
 \E\biggl[\ln\biggl( \frac{Z}{Z_{\B,\SPA}} \biggr) \biggr]
\end{align*}
as functions of $\alpha$ for $n=10$. The horizontal dashed lines are
conjectured asymptotic behavior of
$\ln\bigl( \E[Z] / \E[Z_{\B,\SPA}] \bigr)$ at $\alpha = 0$ and
$\pi$. Although $Z / Z_{\B,\SPA}$ or
$\ln\bigl( Z / Z_{\B,\SPA}\bigr)$ is directly related to the performance of Bethe/SPA
approximation of $|\perm(\btheta)|^2$, it is challenging to characterize their distribution
or statistics directly, and we leave it for future work. In this work, we focus on
$\E[Z] / \E[Z_{\B,\SPA}]$ instead.
Supplementary numerical results can be found in
\iflongversion Appendix~\ref{appx:spa:details}. \else \cite{zhou2026}.\fi


\begin{figure}[t]
 \centering
 \includegraphics[width=\linewidth]{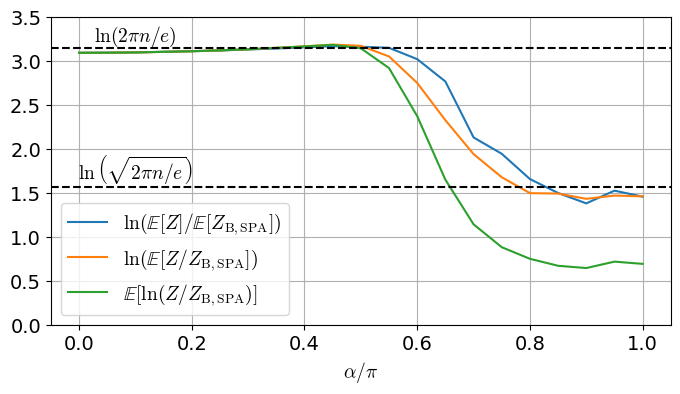}
 \caption{Numerical results of SPA on DE-NFG $\sfnDE(\btheta)$ for
  $n = 10$.}
 \label{fig:DENFG:SPA:alpha}
\end{figure}


Interestingly, even though the SPA messages are initialized to be rank-$2$ PSD
matrices, two types of degeneracies are observed for fixed-point messages and
beliefs. The first type of degeneracy is when $\alpha$ is below a certain
threshold (which depends on $n$): the fixed-point messages are always rank-$1$
matrices. As a result, the ratio $\E[Z] / \E[Z_{\B,\SPA}]$ asymptotically
behave like $\bigl( \E[\perm(\btheta)] / \E[\perm_\B(\btheta)] \big|_{\alpha=0}
\bigr)^2$, and can be studied with the techniques
in~\cite{ngDoublecoverbasedAnalysisBethe2022}.\footnote{The results in
\cite{ngDoublecoverbasedAnalysisBethe2022} can be extended in a straightforward
manner from matrices with non-negative-real-valued entries to matrices with
complex-valued entries.}


The second type of degeneracy is when $\alpha$ is above a certain threshold
(which depends on $n$): in particular, at $\alpha=\pi$, the entries of
$\btheta$ are generated i.i.d. according to the standard complex Gaussian
distribution. In this case, with high probability, the SPA fixed-point messages
and fixed-point beliefs are diagonal matrices. Furthermore, the fixed-point
Bethe free energy $Z_{\B,\SPA}$ is actually approximating
$\perm(|\btheta|^2)=Z(\sfnDE(\mathbf 0_n,|\btheta|^2))$ where
$|\btheta|^2\defeq(|\theta_{i,j}|^2)_{i,j}$, instead of
$|\perm(\btheta)|^2=Z(\sfnDE(\btheta,\mathbf 0_n))$. Although the ratio $\E[Z]
/ \E[Z_{\B,\SPA}]$ asymptotically behaves like $\E[\perm(\btheta)] /
\E[\perm_\B(\btheta)] \big|_{\alpha=0}$, the distribution of $Z / Z_{\B,\SPA}$
becomes more spread out.


\section {Double Covers of NFGs and DE-NFGs}
\label{sec:2cover}


Fix some NFG $\sfn$. For every positive integer $M$, the degree-$M$ Bethe
partition sum is defined to be
\begin{align*}
 Z_{\B,M}(\sfn)
  & \defeq
 \sqrt[M]{\Braket{Z\bigl(\widetilde{\sfn}(\btheta)\bigr)}_{\widetilde{\sfn}\in\widetilde{\mathcal N}_M}},
\end{align*}
where the arithmetic average
$\Braket{\,\cdot\,}_{\widetilde{\sfn}\in\widetilde{\mathcal N}_M}$ is over all
$M$-covers of $\sfn$. (See~\cite{vontobelCountingGraphCovers2013} for the
details of this definition.) In particular, it was shown that
$Z_{\B,M}(\sfn) \big|_{M = 1} = Z(\sfn)$ and
$\lim\sup_{M \to \infty} Z_{\B,M}(\sfn) =
 Z_{\B}(\sfn)$. In\cite{huangCharacterizingBethePartition2020,
 huangGraphcoverbasedCharacterizationBethe2025}, these definitions, and to some
extent, these results, have been extended from NFGs to DE-NFGs.


Because of the observation in~\eqref{eq:Bethe:approximation:ratio:1}, a
particularly important role is played by the case $M = 2$, i.e., double covers
of $\sfn$~\cite{vontobelUnderstandingRatioPartition2025}. In particular, in the
context of NFGs for permanents, the
paper~\cite{ngDoublecoverbasedAnalysisBethe2022} analyzed double covers of
$\sfn(\btheta)$ and showed that
\begin{align*}
 \bigl(\perm_{\B,2}(\btheta)\bigr)^2
  & = \sum_{\sigma_1,\sigma_2\in \mathcal{S}_n}
 \Bigl(
 \prod_{i\in[n]}
 \theta_{i,\sigma_1(i)}
 \theta_{i,\sigma_2(i)}
 \Bigr)
 \cdot
 2^{-c(\sigma_1\sigma_2^{-1})}
\end{align*}
for non-negative-real-valued matrices $\btheta$. Here, $c(\sigma)$ counts the
number of cycles of length at least two in the cycle representation of
$\sigma \in \mathcal{S}_n$.


In this paper, we generalize this result by considering double covers of the
DE-NFG $\sfnDE(\btheta)$ (see \iflongversion Appendix~\ref{appx:double:cover}
\else \cite{zhou2026} \fi for definitions and examples).
Fig.~\ref{fig:DENFG:ZB2:alpha} shows the numerical results for $Z \defeq
Z\bigl( \sfnDE(\btheta) \bigr)$ and $Z_{\B,2} \defeq Z_{\B,2}\bigl(
\sfnDE(\btheta) \bigr)$ over the same ensembles of random matrices considered
in the previous section. In Fig.~\ref{fig:DENFG:ZB2:alpha} we plot the
empirical values for
\begin{align*}
 \ln\biggl( \sqrt{\frac{\E[Z^2]}{\E[Z^2_{\B,2}]}} \biggr),
 \quad
 \ln\biggl( \E\biggl[\frac{Z}{Z_{\B,2}}\biggr] \biggr),
 \quad
 \E\biggl[\ln\biggl( \frac{Z}{Z_{\B,2}} \biggr) \biggr]
\end{align*}
as functions of $\alpha$ for $n=4$. For
$\ln\bigl( \sqrt{ \E[Z^2] / \E[Z^2_{\B,2}]} \bigr)$, the
horizontal dashed lines are conjectured asymptotics at $\alpha=0$ and $\pi$,
and the blue dashed line is the analytic value at each $\alpha$ for $n=4$. We
see that the empirical and analytic results match well.


\begin{figure}[t]
 \centering
 \vspace*{0.05in}
 \includegraphics[width=1\linewidth]{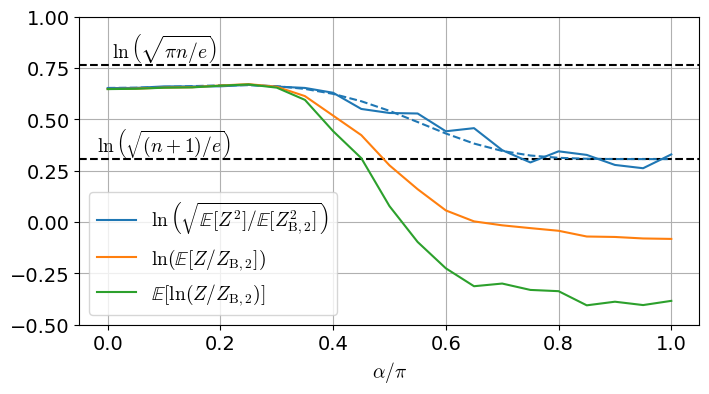}
 \caption{Numerical results of double cover analysis on DE-NFG $\sfnDE(\btheta)$ for $n=4$.}
 \label{fig:DENFG:ZB2:alpha}
\end{figure}


\section{Double Cover Analysis:\\Combinatorics and
  Asymptotics}\label{sec:2cover:counting}


The combinatorics result for $\E[Z^2] / \E[Z^2_{\B,2}]$ relies mainly on the
technique of graph-cover counting developed in
\cite{vontobelAnalysisDoubleCovers2016}. We first describe the general
strategy, which is applicable for any entry-wise i.i.d.~random matrices
$\btheta$. Then we detail two special cases $\btheta$ that we are able to
analyze: the all-one matrix and random complex matrices with zero-mean. The
corresponding asymptotics are consequently obtained by techniques in
\cite{flajoletAnalyticCombinatorics2009}.


The first step of the general strategy is the following decomposition of
$2$-covers of $\sfnDE(\btheta)$.


\begin{prop}\label{prop:double:cover:decomp}
 Any double cover of $\sfnDE(\btheta)$, $\widetilde{\sfn}_{\mathrm{DE}}(\btheta)$, is the
 disjoint union of $\widetilde {\sfn}(\btheta)$ and
 $\widetilde {\sfn}(\overline {\btheta})$ for some unique
 $\widetilde{\sfn}\in \widetilde {\mathcal N}_2$, and hence the partition sum
 is
 $Z\bigl(\widetilde{\sfn}_{\mathrm{DE}}(\btheta)\bigr)=
  \big|Z\bigl(\widetilde{\sfn}(\btheta)\bigr)\big|^2$.
 As a consequence, we have
 \[Z_{\B,2}\bigl(\sfnDE(\btheta)\bigr)=
  \sqrt{\Big\langle\big| Z \bigl(\widetilde{\sfn}(\btheta)\bigr)\big|^2\Big\rangle
   _{\widetilde{\sfn}\in \widetilde {\mathcal N}_2}}.\]
\end{prop}


\begin{proof}
 See \iflongversion Appendix \ref{appx:double:cover}.  \else \cite{zhou2026}.
 \fi
\end{proof}


The second step of the general strategy is to diagonalize the edge weights in a
two-cover $\widetilde{\sfn}(\btheta)$ via NFG transform (see \iflongversion
Appendix \ref{appx:double:cover:counting} \else \cite{zhou2026} \fi or
\cite{vontobelAnalysisDoubleCovers2016,ngDoublecoverbasedAnalysisBethe2022} for
the details). Unlike~\cite[Prop.~1]{ngDoublecoverbasedAnalysisBethe2022}, we do
not consider the superposition of all $2$-covers $\widetilde{\sfn}(\btheta)$ at
this point yet.


\begin{figure}[t]
 \centering
 \begin{minipage}{.23\textwidth}
\centering
\begin{tikzpicture}[on grid, auto]
\tikzstyle{state}=[shape=rectangle, draw, minimum size=4mm]
\pgfmathsetmacro{\hdis}{2.4} 
\pgfmathsetmacro{\vdis}{0.8} 
    \foreach \i in {1,2,3,4,5}{
        \node[state, label=left:$\hatfl{\i}$] (f\i) at (0,-\i*\vdis) {};
        \node[state, label=right:$\hatfr{\i}$] (g\i) at (\hdis,-\i*\vdis) {};
    }
    \foreach \i in {1,2,3,4,5}{
        \foreach \j in {1,2,3,4,5}{                
                 \draw[double] (f\i) -- (g\j) ;
        }
    }
\draw[blue, ultra thick, double] (f1) -- (g1);

\draw[red, ultra thick] (f2) -- (g2);
\draw[red, ultra thick] (f3) -- (g2);
\draw[red, ultra thick] (f3) -- (g3);
\draw[red, ultra thick] (f2) -- (g3);

\draw[green, ultra thick] (f4) -- (g5);
\draw[green, ultra thick] (f5) -- (g5);
\draw[green, ultra thick] (f5) -- (g4);
\draw[green, ultra thick] (f4) -- (g4);
\end{tikzpicture}
\end{minipage}%
 \caption{A valid configuration of NFG $\widehat{N}(\btheta)$. Possible $(1,1)$-edge,
  $(0,1)$-cycle and $(1,0)$-cycle are colored in blue, red and green,
  respectively.
 } \label{fig:doublecover:NFG:cong}
\end{figure}
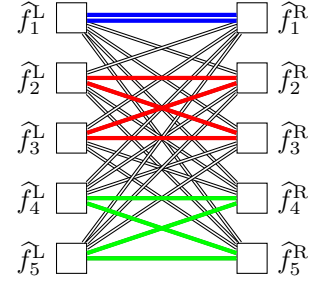


Denote the transformed NFG of $\widetilde{\sfn}(\btheta)$ by
$\widehat{\sfn}(\btheta)$. Correspondingly, function nodes and variables (over
$\mathcal X^2$) in $\widehat{\sfn}(\btheta)$ are denoted as
$\hatfl{i},\hatfr{i},\widehat{f}_{i,j}$ and $\hatx{i,j}$, and the set of all
$2$-covers $\widetilde{\mathcal N}_2$ is mapped to $\widehat{\mathcal N}_2$.
Most importantly, the partition sum is preserved under NFG transform, that is,
$Z\bigl( \widetilde{\sfn}(\btheta) \bigr) = Z\bigl( \widehat{\sfn}(\btheta)
\bigr)$. To find the latter, we shall classify all valid configurations of
$\widehat{\sfn}(\btheta)$, that is, values of $\set{\hatx{i,j}}$ for which the
global function of $\widehat{\sfn}(\btheta)$ is non-zero.


\begin{example}
 The NFG in Fig.~\ref{fig:doublecover:NFG:cong} shows a possible valid
 configuration of $\widehat{\sfn}(\btheta)$ for any
 $\widehat{\sfn}\in\widehat{\mathcal N}_2$ for the case $n=5$. We color an
 edge according to the value of its variable $\hatx{i,j}$.
\end{example}


A key point is that the set of valid configurations is independent of
$\widehat{\sfn}\in\widehat{\mathcal N}_2$, so we denote it by $\validc$. This
enables the following exchange of orders of summation:
\begin{align*}
 \sum_{\widehat{\sfn}\in\widehat{\mathcal N}_2}\Biggl\vert
 \sum_{\hatbfx \in \validc}
 g_{\widehat{\sfn}}(\hatbfx)\Biggr\vert^2
 =\sum_{\substack{\hatbfx[(1)]\in \validc \\\hatbfx[(2)]\in \validc}}
 \sum_{\widehat{\sfn}\in\widehat{\mathcal N}_2}
 g_{\widehat{\sfn}}(\hatbfx[(1)])\overline{g_{\widehat{\sfn}}(\hatbfx[(2)])}.
\end{align*}
In addition, for any valid configuration in
$\validc$, every vertex of
$\widehat{\sfn}(\btheta)$ is one of the three cases:
\begin{itemize}
 \item the endpoint of exactly one $(1,1)$-edge,
 \item vertex in exactly one $(0,1)$-cycle,
 \item vertex in exactly one $(1,0)$-cycle.
\end{itemize}
Furthermore, $\validc$ is related to
$\mathcal{S}_n\times \mathcal{S}_n$ such that any cycle in a valid configuration takes the form
$\bigcup_{i\in C}\bigl\{ (i,\sigma_1(i)),(i,\sigma_2(i)) \bigr\}$, where $C$
is a cycle of length at least $2$ in $\sigma_{1}\sigma_2^{-1}$.
\addtocounter{thm}{-1}
\begin{example} [continued]
 The valid configuration shown in Fig.~\ref{fig:doublecover:NFG:cong} corresponds to
 $4=2^2$ pairs of permutations, where the exponent is the number of cycles in that configuration.
\end{example}


Therefore, the cycle index of the symmetric group $\mathcal{S}_n$ arises
naturally from the counting arguments. It remains to find the contribution of
each valid configuration to $\E \big[Z^2(\sfnDE(\btheta))\big]$ and $ \E
\big[Z_{\B,2}^2(\sfnDE(\btheta))\big]$. For this, we made additional
assumptions concerning $\btheta$, in the next two subsections.


\begin{definition}
 Let $n$ be a positive integer. The cycle index of $\mathcal{S}_n$ is defined
 to be
 $Z_n \defeq \frac{1}{n!}\sum_{\sigma\in
   \mathcal{S}_n}\prod_{k\in[n]}z_k^{c_k(\sigma)}$, where $z_k,k\in [n]$, are
 indeterminates and where $c_k(\sigma)$ counts the number of cycles of length
 $k$ in the cycle representation of $\sigma$. For $n = 0$, let
 $Z_0 \defeq 1$.
 \edefn
\end{definition}


\begin{definition}
 Let $n\in\N$. Define the following functions:
 \begin{alignat*}{3}
  \psi_n : \
   & \R_{\geq 0}^4 &  & \to     &  & \ \R_{\geq 0},                                               \\
   & (a,b,m,c)     &  & \mapsto &  & \ (n!)^2Z_n \text{ with }z_1{=}a,z_k{=}mb^k{+}c^k,k{\geq} 2, \\
  \Psi_n : \
   & \R_{\geq 0}^3 &  & \to     &  & \ \R_{\geq 0},                                               \\
   & (a,b,m)       &  & \mapsto &  & \ (n!)^2Z_n \text{ with }z_1{=}a,z_k{=}mb^k,k{\geq} 2.
 \end{alignat*}
 \edefn
\end{definition}


\subsection{All-one matrix}


We have the following results for all-one matrix of size $n\times n$. Although
the permanent of the all-one matrix is trivial to compute, we conjecture that a
similar asymptotics result of $\E\bigl[ Z^2\bigl( \sfnDE(\btheta) \bigr) \bigr]
/ \E \bigl[ Z_{\B,2}^2\bigl( \sfnDE(\btheta) \bigr) \bigr]$ shall hold for
entry-wise i.i.d.~non-negative-real-valued matrices $\btheta$.


\begin{thm}
 \label{thm:double:cover:allone}

 Let $\btheta=\mathbf 1_n$. It holds that
 \begin{align*}
  Z^2\bigl(\sfnDE(\btheta)\bigr)
   & = (n!)^4,\nonumber                                                       \\
  Z_{\B,2}^2\bigl(\sfnDE(\btheta)\bigr)
   & = \sum_{k=0}^{n} \binom{n}{k}^2 \Psi_{k}(0,1,1/2) \Psi_{n-k}^2(1,1,1/2).
 \end{align*}
\end{thm}


\begin{proof}
 A detailed proof can be found in \iflongversion Appendix
 \ref{appx:double:cover:allone}. \else \cite{zhou2026}. \fi
\end{proof}


\begin{corollary}[of Thm.~\ref{thm:double:cover:allone}]
 \label{cor:double:cover:allone}
 Let $\btheta=\mathbf 1_n$. It holds that
 \begin{align*}
  \frac{Z(\sfnDE(\btheta))}{Z_{\B,2}(\sfnDE(\btheta))}
   & \sim\sqrt{\frac{\pi n}{e}},
 \end{align*}
 where $a(n) \sim b(n)$ stands for
 $\lim_{n \to \infty} \frac{a(n)}{b(n)} = 1$.
\end{corollary}


\begin{proof}
  A detailed proof can be found in \iflongversion Appendix
 \ref{appx:double:cover:allone:asym}. \else \cite{zhou2026}. \fi
\end{proof}


\subsection{Complex matrix with zero-mean}


\begin{assumption}
 \label{ass:iid:matrices}

 Let $\btheta$ be a random complex matrix of size $n\times n$ whose entries are
 i.i.d.~according to some distribution $\mathcal D$ on $\C$ with $\mu_{1,0}=0,
 |\mu_{2,0}|<\mu_{1,1}$ and $\mu_{2,2}<\infty$, where
 $\mu_{p,q}\defeq\E_{\theta\sim \mathcal
 D}\Big[\theta^p\overline{\theta^q}\Big]$ for $p,q\in \N$.
\end{assumption}


\begin{thm}
 \label{thm:double:cover:complex}
 Given Assumption \ref{ass:iid:matrices}, it holds that
 \begin{align*}
  \E \big[Z^2(\sfnDE(\btheta))\big]
   & = \psi_n (\mu_{2,2}, \mu_{1,1}^2, 2, |\mu_{2,0}|^2), \\
  \E \big[Z_{\B,2}^2(\sfnDE(\btheta))\big]
   & = \psi_n (\mu_{2,2}, \mu_{1,1}^2, 1, |\mu_{2,0}|^2).
 \end{align*}
\end{thm}


\begin{proof}
  A detailed proof can be found in \iflongversion Appendix
 \ref{appx:double:cover:complex}. \else \cite{zhou2026}. \fi
\end{proof}


\begin{corollary}[of
  Thm.~\ref{thm:double:cover:complex}]\label{cor:double:cover:complex}
 Given Assumption \ref{ass:iid:matrices}, it holds that
 \begin{align*}
   & \sqrt\frac{\E \big[Z^2(\sfnDE(\btheta))\big]}{\E \big[Z_{\B,2}^2(\sfnDE(\btheta))\big]}
  \sim
  \sqrt{\frac{n+1+C}{e}},
 \end{align*}
 where $C$ is some constant that depends only on the moments of $\mathcal D$.
\end{corollary}


\begin{proof}
 See \iflongversion Appendix \ref{appx:double:cover:complex:asym}
 \else
  \cite{zhou2026}
 \fi
 for the value of $C$ and a detailed proof.
\end{proof}


In the special case where entries of $\btheta$ follow the standard complex
Gaussian distribution, we have $\mu_{1,1}=1,\mu_{2,0}=0$ and $\mu_{2,2}= 2$.
Then the asymptotics given by Cor.~\ref{cor:double:cover:complex} is
$\sqrt{(n+1)/e}$. This matches with our numerical result for $\alpha=\pi$ shown
in Fig.~\ref{fig:DENFG:ZB2:alpha}. For the case \(\mu_{1,1}=|\mu_{2,0}|>0\), see 
Appendix \ref{appx:double:cover:complex:equal}.


\section{Conclusion}


Our numerical experiments showed that the SPA-based Bethe approximation
degrades as the distribution of $\theta$ tends to the standard complex Gaussian
distribution. In the graph-cover analysis, we derived exact expressions and
asymptotics for the second moments of $|\perm(\btheta)|^2$ and its degree-$2$
Bethe approximation for zero-mean complex matrices. The observation in
Eq.~\eqref{eq:Bethe:approximation:ratio:1} only holds when $\arg{\theta}$ is
confined under some threshold, but does not hold when the distribution of
$\theta$ tends to the standard complex Gaussian distribution.

\newpage
\bibliographystyle{IEEEtran}
\bibliography{ref, ref1}

\begin{thebibliography}{10}
\providecommand{\url}[1]{#1}
\csname url@samestyle\endcsname
\providecommand{\newblock}{\relax}
\providecommand{\bibinfo}[2]{#2}
\providecommand{\BIBentrySTDinterwordspacing}{\spaceskip=0pt\relax}
\providecommand{\BIBentryALTinterwordstretchfactor}{4}
\providecommand{\BIBentryALTinterwordspacing}{\spaceskip=\fontdimen2\font plus
\BIBentryALTinterwordstretchfactor\fontdimen3\font minus \fontdimen4\font\relax}
\providecommand{\BIBforeignlanguage}[2]{{%
\expandafter\ifx\csname l@#1\endcsname\relax
\typeout{** WARNING: IEEEtran.bst: No hyphenation pattern has been}%
\typeout{** loaded for the language `#1'. Using the pattern for}%
\typeout{** the default language instead.}%
\else
\language=\csname l@#1\endcsname
\fi
#2}}
\providecommand{\BIBdecl}{\relax}
\BIBdecl

\bibitem{vontobelBethePermanentNonnegative2013}
P.~O. Vontobel, ``The {{Bethe}} permanent of a nonnegative matrix,'' \emph{IEEE Trans. Inf. Theory}, vol.~59, no.~3, pp. 1866--1901, Mar. 2013.

\bibitem{aaronsonComputationalComplexityLinear2013}
S.~Aaronson and A.~Arkhipov, ``The computational complexity of linear optics,'' \emph{Theory Comput.}, vol.~9, no.~4, pp. 143--252, Feb. 2013.

\bibitem{aaronsonBosonSamplingFarUniform2013}
------, ``{{BosonSampling}} is far from uniform,'' Oct. 2013, [Online]. Available: http://arxiv.org/abs/1309.7460.

\bibitem{eldarApproximatingPermanentRandom2018}
L.~Eldar and S.~Mehraban, ``Approximating the permanent of a random matrix with vanishing mean,'' in \emph{2018 {{IEEE}} 59th {{Annu}}. {{Symp}}. {{Found}}. {{Comput}}. {{Sci}}. {{FOCS}}}, Oct. 2018, pp. 23--34.

\bibitem{jiApproximatingPermanentRandom2021}
Z.~Ji, Z.~Jin, and P.~Lu, ``Approximating permanent of random matrices with vanishing mean: {{Made}} better and simpler,'' in \emph{Proc. {{Thirty-Second Annu}}. {{ACM-SIAM Symp}}. {{Discrete Algorithms}}}, ser. {{SODA}} '21.\hskip 1em plus 0.5em minus 0.4em\relax USA: {Society for Industrial and Applied Mathematics}, Mar. 2021, pp. 959--975.

\bibitem{nezamiPermanentRandomMatrices2021}
S.~Nezami, ``Permanent of random matrices from representation theory: {{Moments}}, numerics, concentration, and comments on hardness of {{Boson-sampling}},'' Apr. 2021, [Online]. Available: http://arxiv.org/abs/2104.06423.

\bibitem{chabaudResourcesBosonicQuantum2023}
U.~Chabaud and M.~Walschaers, ``Resources for {{Bosonic}} quantum computational advantage,'' \emph{Phys. Rev. Lett.}, vol. 130, no.~9, p. 090602, Mar. 2023.

\bibitem{forneyCodesGraphsNormal2001}
G.~Forney, ``Codes on graphs: {{Normal}} realizations,'' \emph{IEEE Trans. Inf. Theory}, vol.~47, no.~2, pp. 520--548, Feb. 2001.

\bibitem{loeligerIntroductionFactorGraphs2004}
H.-A. Loeliger, ``An introduction to factor graphs,'' \emph{IEEE Sig. Process. Mag.}, vol.~21, no.~1, pp. 28--41, 2004.

\bibitem{kschischangFactorGraphsSumproduct2001}
F.~Kschischang, B.~Frey, and H.-A. Loeliger, ``Factor graphs and the sum-product algorithm,'' \emph{IEEE Trans. Inform. Theory}, vol.~47, no.~2, pp. 498--519, 2001.

\bibitem{wainwrightGraphicalModelsExponential2008}
M.~J. Wainwright and M.~I. Jordan, ``Graphical models, exponential families, and variational inference,'' \emph{Found. Trends\textregistered{} Mach. Learn.}, vol.~1, no. 1--2, pp. 1--305, Nov. 2008.

\bibitem{yedidiaConstructingFreeenergyApproximations2005}
J.~Yedidia, W.~Freeman, and Y.~Weiss, ``Constructing free-energy approximations and generalized belief propagation algorithms,'' \emph{IEEE Trans. Inf. Theory}, vol.~51, no.~7, pp. 2282--2312, Jul. 2005.

\bibitem{vontobelCountingGraphCovers2013}
P.~O. Vontobel, ``Counting in graph covers: {{A}} combinatorial characterization of the {{Bethe}} entropy function,'' \emph{IEEE Trans. Inf. Theory}, vol.~59, no.~9, pp. 6018--6048, Sep. 2013.

\bibitem{ngDoublecoverbasedAnalysisBethe2022}
K.~S. NG and P.~O. Vontobel, ``Double-cover-based analysis of the {{Bethe}} permanent of non-negative matrices,'' in \emph{Proc. {{IEEE Inf}}. {{Theory Workshop}}}, Nov. 2022, pp. 672--677.

\bibitem{huangDegreeMBetheSinkhorn2024}
Y.~Huang, N.~Kashyap, and P.~O. Vontobel, ``Degree-{{M Bethe}} and {{Sinkhorn}} permanent based bounds on the permanent of a non-negative matrix,'' \emph{IEEE Trans. Inf. Theory}, vol.~70, no.~7, pp. 5289--5308, Jul. 2024.

\bibitem{vontobelUnderstandingRatioPartition2025}
P.~O. Vontobel, ``Understanding the ratio of the partition sum to its {{Bethe}} approximation via double covers,'' in \emph{Proc. {{IEEE Inf}}. {{Theory Workshop}}}, Sep. 2025, pp. 1--6.

\bibitem{caoDoubleedgeFactorGraphs2017}
M.~X. Cao and P.~O. Vontobel, ``Double-edge factor graphs: {{Definition}}, properties, and examples,'' in \emph{Proc. {{IEEE Inf}}. {{Theory Workshop}}}, Nov. 2017, pp. 136--140.

\bibitem{huangCharacterizingBethePartition2020}
Y.~Huang and P.~O. Vontobel, ``Characterizing the {{Bethe}} partition function of double-edge factor graphs via graph covers,'' in \emph{Proc. {{IEEE Int}}. {{Symp}}. {{Inf}}. {{Theory}}}, Jun. 2020, pp. 1331--1336.

\bibitem{huangGraphcoverbasedCharacterizationBethe2025}
------, ``Graph-cover-based characterization of the {{Bethe}} partition function of double-edge factor graphs,'' Jun. 2025, [Online]. Available: http://arxiv.org/abs/2506.16250.

\bibitem{kreutz-delgadoComplexGradientOperator2009}
K.~{Kreutz-Delgado}, ``The complex gradient operator and the {{CR-calculus}},'' Jun. 2009, [Online]. Available: http://arxiv.org/abs/0906.4835.

\bibitem{vontobelAnalysisDoubleCovers2016}
P.~O. Vontobel, ``Analysis of double covers of factor graphs,'' in \emph{2016 {{Int}}. {{Conf}}. {{Signal Process}}. {{Commun}}. {{SPCOM}}}, Jun. 2016, pp. 1--5.

\bibitem{flajoletAnalyticCombinatorics2009}
P.~Flajolet and R.~Sedgewick, \emph{Analytic {{Combinatorics}}}, 1st~ed.\hskip 1em plus 0.5em minus 0.4em\relax Cambridge University Press, Jan. 2009.

\bibitem{loeligerQuantumMeasurementMarginalization2020}
H.-A. Loeliger and P.~O. Vontobel, ``Quantum measurement as marginalization and nested quantum systems,'' \emph{IEEE Trans. Inf. Theory}, vol.~66, no.~6, pp. 3485--3499, Jun. 2020.

\end{thebibliography}

\iflongversion
 \onecolumn
 \appendices


 \section{Entrywise i.i.d.~Complex Matrices}\label{appx:random:matrix}
 We consider ensembles of complex-valued matrices $\btheta$ whose entries are
 i.i.d.~complex random variables.
 \begin{definition}
  Let $\mathcal D_\alpha,\alpha\in[0,\pi]$, be the distribution over
  \[S_\alpha=
   \bigl\{
   r e^{\iota \phi}
   \bigm|
   r \ge 0,\ \phi \in [-\alpha,\alpha]
   \bigr\},\]
  with the probability density function given by
  \[
   p_\alpha(\theta)=p_\alpha(r,\phi)\defeq
   \begin{cases}
   2r\exp(-r^2),                     & \alpha=0    \\
   r\exp(-r^2)\cdot\frac{1}{\alpha}, & \alpha\neq0
   \end{cases},\quad \theta=re^{\iota\phi}\in S_\alpha.
  \]
  \edefn
 \end{definition}

 Note that the distribution $\mathcal D_0$ is the Rayleigh distribution over
 $\R_{\geq 0}$ with scale $1/\sqrt{2}$, and $\mathcal D_\pi$ is the standard
 complex Gaussian distribution $\mathcal {CN}(0,1)$. So as $\alpha$ goes from
 $0$ to $\pi$, the distribution $D_\alpha$ gradually becomes a
 ``fully-complex'' distribution, and we study how the ratios ${Z}/{Z_{\B}}$ and
 $ {Z}/{Z_{\B,2}}$ for the DE-NFG $\sfnDE(\btheta)$ change in this process.
 
 \begin{prop}
  Let $\mu_{p,q}(\alpha)\defeq\E_{\theta\sim \mathcal
    D_\alpha}\Big[\theta^p\overline{\theta^q}\Big]$ for $p,q\in \N$. We have
  \begin{align*}
   \mu_{1,0}(\alpha) & =
   \begin{cases}
   \frac{\sqrt \pi}{2},                  & \alpha=0     \\
   \frac{\sqrt \pi\sin \alpha}{2\alpha}, & \alpha\neq 0
   \end{cases},  \\
   \mu_{1,1}(\alpha) & =1,                               \\
   \mu_{2,0}(\alpha) & =
   \begin{cases}
   1,                             & \alpha=0     \\
   \frac{\sin(2\alpha)}{2\alpha}, & \alpha\neq 0
   \end{cases},         \\
   \mu_{2,1}(\alpha) & =
   \begin{cases}
   \frac{3\sqrt \pi}{4},                  & \alpha=0     \\
   \frac{3\sqrt \pi\sin \alpha}{4\alpha}, & \alpha\neq 0
   \end{cases}, \\
   \mu_{2,2}(\alpha) & =2.
  \end{align*}

 \end{prop}


 \section{Bethe Free Energy Function of the DE-NFG \(\sfnDE(\btheta)\) }\label{appx:Bethe:details}
 \begin{definition}[Belief polytope of $\sfnDE(\btheta)$]\label{def:DENFG:beliefpolytope}
  Consider a complex matrix $\btheta$ and the associated DE-NFG $\sfnDE(\btheta)$. We let
  \[\bbeta\defeq \bigl((\bl{i})_{i\in[n]},(\br{j})_{j\in[n]},(\be{i,j})_{i,j\in [n]^2}\bigr)\]
  be a collection of complex matrices based on the matrices
  \begin{align*}
   \bl{i} \defeq \bigl(\bl{i}(\ell,k) \bigr)_{(\ell,k)\in [n]^2},\quad
   \br{j} \defeq \bigl(\br{j}(\ell,k) \bigr)_{(\ell,k)\in [n]^2},\quad
   \be{i,j} \defeq \bigl(\be{i,j}(x,x') \bigr)_{(x,x')\in\mathcal{X}^2}.
  \end{align*}
  The belief polytope $\mathcal B$ is defined to be the set
  \begin{align*}
   \mathcal B =
   \left\{ \bbeta\ \middle\vert
   \begin{array}{c}
    \tr(\mathbf{1}_{n}\bl{i})=1, \quad \bl{i}\succeq 0, \quad \forall i\in[n];           \\
    \tr(\mathbf{1}_{n}\br{j})=1, \quad \br{j}\succeq 0, \quad \forall j\in[n];           \\
    \tr(\mathbf{1}_{2}\be{i,j})=1, \quad \be{i,j}\succeq 0, \quad \forall (i,j)\in[n]^2; \\[0.5em]
    \be{i,j}=\begin{bmatrix}\sum_{k\neq j,\ell\neq j}\bl{i}(\ell,k)&\sum_{\ell\neq j}\bl{i}(\ell,j)\\\sum_{k\neq j}\bl{i}(j,k)&\bl{i}(j,j)\end{bmatrix},
    \quad \forall (i,j)\in[n]^2;                                                         \\[1.5em]
    \be{i,j}=\begin{bmatrix}\sum_{k\neq i,\ell\neq i}\br{j}(\ell,k)&\sum_{\ell\neq i}\br{j}(\ell,i)\\\sum_{k\neq i}\br{j}(i,k)&\br{j}(i,i)\end{bmatrix},
    \quad \forall (i,j)\in[n]^2
   \end{array}\right\}
  \end{align*}
  where $\bbeta\in\mathcal B$ is called a pseudo-belief.
  \edefn
 \end{definition}
 \begin{remark}
  We make some remarks about the above definition.
  \begin{itemize}
   \item The belief $\bl{i}$ is associated with the set of valid configurations
   of $\tfl{i}$. In particular, $\bl{i}(\ell,k)$ is associated with the valid
   configuration $(\xl{i,j})_j=\mathbf u_\ell,(\xlprime{i,j})_j=\mathbf u_k$,
   where $\mathbf u_\ell$ is the $\ell$-th standard unit vector in $\R^n$.
   \item The first three constraints in $\mathcal B$ ensure that each matrix of
   a pseudo-belief is a valid quantum mass function
   \cite[Def.~1]{loeligerQuantumMeasurementMarginalization2020}, which is a
   quantum generalization of the probability mass function (used for defining
   the belief polytope of $\sfn(\btheta)$
   \cite[Def.~9]{vontobelBethePermanentNonnegative2013}).
   \item The last two constraints in $\mathcal B$ are referred as the ``edge
   consistency'' constraints. They essentially require that the marginal of
   $\bl{i}$ with respect to $\txl{i,j}$ is the same as $\be{i,j}$, and
   similarly for $\br{j}$. An alternative formalism is provided in the lemma
   below.
   \item Strictly speaking, we shall also introduce belief matrices for
   variables $\txl{i,j}$ and $\txr{i,j}$ (see Fig.~\ref{fig:DENFGperm}). Since
   the function \(\tfe{i,j}\) imposes an equality constraint on $\txl{i,j}$ and
   $\txr{i,j}$, the consistency constraint requires these belief matrices to be
   the same as $\be{i,j}$, and we omit them in the above definition for
   simplicity.
  \end{itemize}
 \end{remark}

 \begin{lem}
  Let $\mathbf V_i,i\in[n]$, be the matrix given by
  \[
   \mathbf V_i=
   \begin{bNiceArray}{ccccccc}[last-row]
    1 & \cdots & 1 & 0 & 1 & \cdots & 1 \\
    0 & \cdots & 0 & 1 & 0 & \cdots & 0 \\
    \rule[15pt]{0pt}{0pt}
    & i-1 & & & & n-i & \\
    \CodeAfter
    \UnderBrace[shorten, yshift=3pt]{2-1}{2-3}{}
    \UnderBrace[shorten, yshift=3pt]{2-5}{2-7}{}
   \end{bNiceArray}
   .\]
  Then the ``edge consistency'' constraints in defining $\mathcal B$ can be
  written as
  \[
   \left\{
   \begin{aligned}
    \mathbf V_j \bl{i} \mathbf V_j^\T & =\be{i,j} \\
    \mathbf V_i \br{j} \mathbf V_i^\T & =\be{i,j}
   \end{aligned}\right., \quad \forall (i,j)\in [n]^2.
  \]
 \end{lem}

 To generalize the Bethe free energy function to the complex-valued belief
 $\bbeta$, we must take care of the multivalued complex logarithm. In the
 following, we use the principal value of the complex logarithm, that is,
 \[
  \ln \theta \defeq\ln|\theta|+\iota\phi,\quad \theta\in\C\setminus\R_{\leq 0},
 \]
 where $\theta=|\theta|e^{\iota \phi}$ for $\phi\in(-\pi,\pi)$. This choice
 allows the identity $\ln\overline\phi=\overline{\ln\phi}$ to hold, and the
 functions defined below are always real valued. We also adopt the convention
 that $0\ln 0=0$.

 \begin{definition}
  The Bethe free energy function associated with the DE-NFG $\sfnDE(\btheta)$ is defined to be the function
  \[
   F_{\B}:\mathcal B\to\R,\quad \bbeta\mapsto  U_{\B}(\bbeta) - H_{\mathrm{B}}(\bbeta)
  \]
  where
  \begin{align*}
    & U_{\B} : \mathcal B\to \R,\quad \bbeta \mapsto \sum_{i,j} U_{\B,i,j}(\be{i,j})                                 \\
    & H_{\B}: \mathcal B \to \R,\quad \bbeta \mapsto \sum_iH_{\B,i}\big(\bl{i}\big)+\sum_{j}H_{\B,j}\big(\br{j}\big)
   -\sum_{i,j}H_{\B,i,j}(\be{i,j})
  \end{align*}
  with
  \begin{alignat*}{3}
   U_{\B,i,j}:\  & \be{i,j} &  & \mapsto &  & -\sum_{x,x'} \be{i,j}(x,x')\ln\big(\mathbf W_{i,j}(x,x')\big) \\
   H_{\B,i}:\    & \bl{i}   &  & \mapsto &  & -\sum_{\ell,k}\bl{i}(\ell,k)\ln\big(\bl{i}(\ell,k)\big)       \\
   H_{\B,j}:\    & \br{j}   &  & \mapsto &  & -\sum_{\ell,k}\br{j}(\ell,k)\ln\big(\br{j}(\ell,k)\big)       \\
   H_{\B,i,j}:\  & \be{i,j} &  & \mapsto &  & -\sum_{x,x'} \be{i,j}(x,x')\ln\big(\be{i,j}(x,x')\big)
  \end{alignat*}
  Here, $U_{\B}$ and $H_{\B}$ are called the Bethe average energy function and the Bethe entropy function, respectively.
  \edefn
 \end{definition}

 \begin{definition}
  Let $\btheta$ be a complex matrix and $\sfnDE(\btheta)$ be the associated DE-NFG. Define the Bethe partition function $Z_\B$ of $\sfnDE(\btheta)$ as
  \begin{align*}
   Z_\B (\sfnDE(\btheta)) \defeq \exp\Big(- \stat_{\bbeta \in \mathcal B} F_\B(\bbeta)\Big),
  \end{align*}
  where the operator $\stat$ returns the set of values of the function where the function is stationary.
  Here, we assume that the stationary point of \(F_\B(\,\cdot\,)\) exists and is unique.
  \edefn
 \end{definition}

 We look for stationary points of $F_\B$ instead of the minimum for the
 following reasons. First, the Bethe free energy function (or any concerned
 objective function) could be a complex-valued function for a general DE-NFG,
 so only stationary points make sense. Second, throughout our investigation,
 the Bethe free energy function for the DE-NFG \(\sfnDE(\btheta)\) does not
 exhibit the desired convexity properties, in contrast to its classical
 counterparts in \cite{vontobelBethePermanentNonnegative2013}.


 \section{Details of SPA on the DE-NFG \(\sfnDE(\btheta)\)}\label{appx:spa:details}

 \subsection{Message initialization}
 Consider $\mathbf M=\mathbf U \mathbf D \mathbf U^\mathrm H$ where $\mathbf U$
 is a Haar-random matrix over the unitary group $\mathrm U(2)$, and $\mathbf
 D=\mathrm{diag}(D_0,D_1),D_i\overset{\mathrm {i.i.d.}}{\sim}\mathrm{Exp}(1)$.
 All the left-going messages in a DE-NFG are initialized i.i.d.~as $\mathbf
 M/M_{00}$. The normalization by $M_{00}$ enables a simplified implementation
 of the message update rules.

 \subsection{Message update rules}
 \begin{figure}[t]
  \centering
  \begin{tikzpicture}[on grid, auto]
\tikzstyle{state}=[shape=rectangle, draw, minimum size=4mm]
\pgfmathsetmacro{\hdis}{8} 
\pgfmathsetmacro{\vdis}{1} 
\node[state, label=left:$\tfl{i}$] (f) at (0,-2*\vdis) {};
\node[state, label=right:$\tfr{j}$] (g) at (\hdis,-0*\vdis) {};
\node[state, label=above left:$\tfe{i,j}$] (m) at (0.5*\hdis, -1*\vdis) {};
\draw[double] (f) -- (m) node [pos=0.5, above, sloped] {$\txl{i,j}$} 
node [pos=0.2,sloped] {\color{red}$\muleft{L}{i,j,t}$} node [pos=0.8,below, sloped] {\color{blue}$\muright{L}{i,j,t}$};
\draw[double] (m) -- (g) node [pos=0.5, below, sloped] {$\txr{i,j}$} 
node [pos=0.2,sloped] {\color{red}$\muleft{R}{i,j,t}$} node [pos=0.8,below,sloped] {\color{blue}$\muright{R}{i,j,t}$}; 
\foreach \i in {0.33,0.66}{
    \draw[double] (f) -- ++ (0.25*\hdis, -\i*\vdis/2);
    \draw[dashed, double] ($(f)+(0.25*\hdis, -\i*\vdis/2)$) -- ++ (0.25*\hdis, -\i*\vdis/2);
    \draw[double] (g) -- ++ (-0.25*\hdis, \i*\vdis/2);
    \draw[dashed, double] ($(g)+(-0.25*\hdis, \i*\vdis/2)$) -- ++ (-0.25*\hdis, \i*\vdis/2);
}
\end{tikzpicture}
  \caption{Zoom in on parts of $\sfnDE(\btheta)$ with SPA messages.}\label{fig:DENFG:SPA}
 \end{figure}
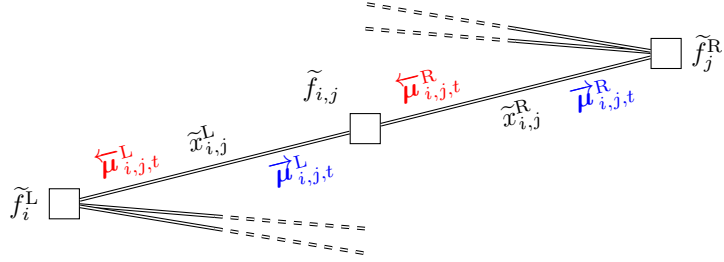

 At each time $t\in\N$, we associate a left-going message
 $\muleft{L}{i,j,t}:\mathcal X^2\to\C$ and a right-going message
 $\muright{L}{i,j,t}:\mathcal X^2\to\C$ with the variable $\txl{i,j}$, and
 similarly for $\txr{i,j}$ (see Fig.~\ref{fig:DENFG:SPA}). The message update
 rules for the DE-NFG $\sfnDE(\btheta,\boldsymbol{\epsilon})$ are derived from
 the sum-product algorithms on general NFGs
 \cite{loeligerIntroductionFactorGraphs2004}. Assume that for $t=0$, all the
 messages are suitably initialized. One can show that the message update rules
 for the function nodes $\tfl{i}$ and $\tfe{i,j}$ for $t\geq 1, (i,j)\in[n]^2$
 are given by, respectively,
 \begin{equation*}
  \begin{split}
    & \left\{\begin{split}
              \muright{L}{i,j,t}(1,1) & =\prod_{j'\neq j}\muleft{L}{i,j',t-1}(0,0)                                  \\
              \frac{\muright{L}{i,j,t}(0,1)}{\muright{L}{i,j,t}(1,1)}
                                      & =\sum_{j'\neq j}\frac{\muleft{L}{i,j',t-1}(1,0)}{\muleft{L}{i,j',t-1}(0,0)} \\
              \frac{\muright{L}{i,j,t}(1,0)}{\muright{L}{i,j,t}(1,1)}
                                      & =\sum_{j'\neq j}\frac{\muleft{L}{i,j',t-1}(0,1)}{\muleft{L}{i,j',t-1}(0,0)} \\
              \frac{\muright{L}{i,j,t}(0,0)}{\muright{L}{i,j,t}(1,1)}
                                      & =\sum_{j'\neq j}\frac{\muleft{L}{i,j',t-1}(1,1)}{\muleft{L}{i,j',t-1}(0,0)}
              +\sum_{\substack{j_1,j_2\neq j                                                                        \\j_1\neq j_2}}\frac{\muleft{L}{i,j_1,t-1}(1,0)}{\muleft{L}{i,j_1,t-1}
               (0,0)}\frac{\muleft{L}{i,j_2,t-1}(0,1)}{\muleft{L}{i,j_2,t-1}(0,0)}
             \end{split}\right. \\
    & \Bigg\{\begin{split}
              \muleft{L}{i,j,t}  & =\mathbf W_{i,j}\odot \muleft{R}{i,j,t}  \\
              \muright{R}{i,j,t} & =\mathbf W_{i,j}\odot \muright{L}{i,j,t}
             \end{split},
  \end{split}
 \end{equation*}
 where $\odot$ denotes the entry-wise product of two matrices.
 The update rules for the function nodes $\tfr{j}$ are similar to $\tfl{i}$, except that
 the incoming messages are \(\muright{R}{i,j,t}\) and the outgoing messages are \(\muleft{R}{i,j,t}\).
 \begin{lem}
  Viewing $\muleft{L}{i,j,t}$ as a $2\times 2$ matrix, we define
  $\overleftarrow\Delta^{\mathrm{L}}_{i,j,t}\defeq\det\big(\muleft{L}{i,j,t}\big)$,
  and similarly for other messages. Then the fourth update rule for $\tfl{i}$ can be replaced by
  \[
   \frac{\overrightarrow{\Delta}^{\mathrm{L}}_{i,j,t}}{\big(\muright{L}{i,j,t}(1,1)\big)^2}
   =\sum_{j'\neq j}\frac{\overleftarrow{\Delta}^{\mathrm{L}}_{i,j',t-1}}{\big(\muleft{L}{i,j',t-1}(0,0)\big)^2}.
  \]

  In the case where \(\mathbf W_{i,j}= \Bigl[
   \begin{smallmatrix}
    1            & \overline{\theta_{i,j}} \\
    \theta_{i,j} & |\theta_{i,j}|^2
   \end{smallmatrix}
   \Bigr]\), which is rank-$1$, the update rules for $\tfe{i,j}$ can be replaced by
  \begin{equation*}
   \left\{
   \begin{aligned}
    \muleft{L}{i,j,t}(0,0)                    & = \muleft{R}{i,j,t}(0,0)                        \\
    \muleft{L}{i,j,t}(0,1)                    & = \overline{\theta_{i,j}}\muleft{R}{i,j,t}(0,1) \\
    \muleft{L}{i,j,t}(1,0)                    & = {\theta_{i,j}}\muleft{R}{i,j,t}(1,0)          \\
    \overleftarrow{\Delta}^\mathrm{L}_{i,j,t} & =
    |\theta_{i,j}|^2 \overleftarrow{\Delta}^\mathrm{R}_{i,j,t}
   \end{aligned}\right. \ \mathrm{and}\
   \left\{
   \begin{aligned}
    \muright{R}{i,j,t}(0,0)                    & =\muright{L}{i,j,t}(0,0)                        \\
    \muright{R}{i,j,t}(0,1)                    & =\overline{\theta_{i,j}}\muright{L}{i,j,t}(0,1) \\
    \muright{R}{i,j,t}(1,0)                    & ={\theta_{i,j}}\muright{L}{i,j,t}(1,0)          \\
    \overrightarrow{\Delta}^\mathrm{R}_{i,j,t} & =
    |\theta_{i,j}|^2 \overrightarrow{\Delta}^\mathrm{L}_{i,j,t}
   \end{aligned}\right. .
  \end{equation*}
  Therefore, we can use $\big(\muleft{L}{i,j,t}(0,0),\muleft{L}{i,j,t}(0,1),\muleft{L}{i,j,t}(1,0),\overleftarrow{\Delta}^{\mathrm{L}}_{i,j,t}\big)$
  to represent the left-going message $\muleft{L}{i,j,t}$, and similarly for other messages.
 \end{lem}

 \subsection{SPA fixed-point belief}
 The SPA-based belief at function nodes $\tfl{i},\tfr{j}$ and $\tfe{i,j}$ at
 time $t$ are given by, respectively,
 \begin{align*}
  \bl{i,t}(\ell,k) & \propto \prod_{j}\muleft{L}{i,j,t}\big([j=\ell],[j=k]\big),(\ell,k)\in[n]^2  \\
  \br{j,t}(\ell,k) & \propto \prod_{i}\muright{R}{i,j,t}\big([i=\ell],[i=k]\big),(\ell,k)\in[n]^2 \\
  \be{i,j,t}       & \propto \mathbf W_{i,j}\odot\muright{L}{i,j,t}\odot\muleft{R}{i,j,t},
 \end{align*}
 and they are normalized according to the sum constraints, i.e.,
 \[
  \tr(\mathbf{1}_{n}\bl{i,t})=1, \quad
  \tr(\mathbf{1}_{n}\br{j,t})=1, \quad \text{and }
  \tr(\mathbf{1}_{2}\be{i,j,t})=1.
 \]

 \begin{prop}[See \protect\cite{caoDoubleedgeFactorGraphs2017}]
  At the fixed point of the SPA, the fixed-point belief vector $\bbeta_{\mathrm F}$ satisfies the edge consistency
  constraints. Furthermore, if the SPA is initialized with positive semidefinite messages, then
  all fixed-point beliefs are positive semidefinite.
 \end{prop}
 As a direct consequence of the above proposition, a fixed point of the SPA corresponds to
 a pseudo-belief in the belief polytope. It turns out that fixed points of the SPA have
 a deeper connection with stationary points of the Bethe free energy function.

 \begin{thm} \label{thm:stationary:fixedpoint}
  Let $\btheta$ be a complex matrix that contains no $0$, and $\sfnDE(\btheta)$ be the
  associated DE-NFG. The interior stationary points of the Bethe free energy function $F_\B$
  correspond to non-diagonal fixed points of the SPA and vice versa.
 \end{thm}

 \begin{proof}
  The idea of the proof of Theorem \ref{thm:stationary:fixedpoint} is the same as
  \cite[Thm.~2]{yedidiaConstructingFreeenergyApproximations2005}, and we use the complex derivatives
  defined in the Wirtinger calculus (see, e.g.,\cite{kreutz-delgadoComplexGradientOperator2009}).

  We associate Lagrangian multipliers with the equality constraints defining
  $\mathcal B$ and form the Lagrangian $\mathcal{L}$, i.e.,
  \begin{align*}
   \mathcal L & = F_\B+\sum_{i,j}\gamma_{i,j} \Bigl[\tr \big(\mathbf{1}_{2}\be{i,j}\big)-1\Bigr]    \\
              & \quad +\sum_i \gamma^{\mathrm L}_i \Bigl[\tr \big(\mathbf{1}_{n}\bl{i}\big)-1\Bigr]
   +\sum_j \gamma^{\mathrm R}_j \Bigl[\tr \big(\mathbf{1}_{n}\br{j}\big)-1\Bigr]                    \\
              & \quad +\sum_{i,j}\sum_{x,x'}\mathrm{Re}
   \Bigl\{\boldsymbol{\lambda}^{\mathrm L}_{i,j}(x,x')
   \Bigl(\be{i,j}(x,x')-\bigl[\mathbf V_j\bl{i}\mathbf V_j^\T\bigr]_{x,x'}\Bigr)\Bigr\}             \\
              & \quad +\sum_{i,j}\sum_{x,x'}\mathrm{Re}
   \Bigl\{\boldsymbol{\lambda}^{\mathrm R}_{i,j}(x,x')
   \Bigl(\be{i,j}(x,x')-\bigl[\mathbf V_i\br{j}\mathbf V_i^\T\bigr]_{x,x'}\Bigr)\Bigr\},
  \end{align*}
  where $\gamma^{\mathrm L}_i,\gamma^{\mathrm R}_j,\gamma_{i,j}\in\R$ and $\boldsymbol{\lambda}^{\mathrm L}_{i,j},\boldsymbol{\lambda}^{\mathrm R}_{i,j}\in\C^{2\times 2}$.
  Since only interior stationary points are considered, all beliefs are positive definite.

  By setting the complex derivatives $\frac{\partial \mathcal L}{\partial
  \bl{i}(\ell,k)}, \frac{\partial \mathcal L}{\partial \br{j}(\ell,k)}$ and
  $\frac{\partial \mathcal L}{\partial \be{i,j}(x,x')}$ to be zero, we obtain
  the following stationarity conditions:
  \begin{equation*}
   \begin{split}
    \be{i,j}(x,x')
    = & \,\big(\mathbf W_{i,j}(x,x')\big)^{-1}\exp\Big\{\frac{1}{2}\big(\boldsymbol{\lambda}^{\mathrm L}_{i,j}(x,x')+\overline{\boldsymbol{\lambda}^{\mathrm L}_{i,j}(x',x)}
    + \boldsymbol{\lambda}^{\mathrm R}_{i,j}(x,x')+\overline{\boldsymbol{\lambda}^{\mathrm R}_{i,j}(x',x)}\big)-1+\gamma_{i,j}\Big\}                                         \\
    \ln \bl{i}(\ell,k)
    = & \,
    \left\{\begin{aligned}
             & \sum_{j\neq \ell}\boldsymbol{\lambda}^{\mathrm L}_{i,j}(0,0)+\boldsymbol{\lambda}^{\mathrm L}_{i,\ell}(1,1)-1-\gamma^{\mathrm L}_i,
             & \ell=k                                                                                                                              \\
             & \sum_{j\neq \ell,k}\boldsymbol{\lambda}^{\mathrm L}_{i,j}(0,0)
            +\frac{1}{2}\big(\boldsymbol{\lambda}^{\mathrm L}_{i,\ell}(1,0)+\overline{\boldsymbol{\lambda}^{\mathrm L}_{i,\ell}(0,1)}
            + \boldsymbol{\lambda}^{\mathrm L}_{i,k}(0,1)+\overline{\boldsymbol{\lambda}^{\mathrm R}_{i,k}(1,0)}\big)-1-\gamma^{\mathrm L}_i,
             & \ell\neq k
           \end{aligned}\right.
   \end{split}
  \end{equation*}
  and similarly for $\br{j}$.

  Finally, by identifying the Lagrangian multipliers with fixed-point messages
  of the SPA, i.e.,
  \[
   \left\{
   \begin{aligned}
    \boldsymbol{\lambda}^{\mathrm L}_{i,j}(x,x') & = \ln \muleft{L}{i,j}(x,x')  \\
    \boldsymbol{\lambda}^{\mathrm R}_{i,j}(x,x') & = \ln \muright{R}{i,j}(x,x')
   \end{aligned}\right. , \quad \forall\, (x,x')\in\mathcal X^2,(i,j)\in[n]^2,
  \]
  we obtain the fixed-point beliefs from the stationarity conditions. (Note that
  we require that the off-diagonal entry of each SPA message be non-zero, so
  that the above quantities are well-defined.) The above shows that an interior
  stationary point of the Bethe free energy function is a fixed point of SPA.
  The reverse direction is the same as in the proof of
  \cite[Thm.~2]{yedidiaConstructingFreeenergyApproximations2005}.
 \end{proof}

 \subsection{Supplementary numerical results}
 In Fig.~\ref{fig:DENFG:SPA:alpha:supp}, we plot the empirical mean and
 standard deviation of
 \begin{align*}
  \ln
  \biggl( \frac{\E[Z]}{\E[Z_{\B,\SPA}]} \biggr),
  \quad
  \ln
  \biggl(
  \E\biggl[\frac{Z}{Z_{\B,\SPA}}\biggr]
  \biggr),
  \quad
  \E\biggl[\ln\biggl( \frac{Z}{Z_{\B,\SPA}} \biggr) \biggr]
 \end{align*}
 as functions of $\alpha$ for $n= 4$ and $10$, respectively. We observe that, as $\alpha$ increases,
 the (empirical) standard deviation of
 \[
  \frac{Z}{Z_{\B,\SPA}}
  \text{ and }
  \ln\biggl( \frac{Z}{Z_{\B,\SPA}} \biggr)
 \]
 increases, indicating that the quality of the SPA-based Bethe approximation
 degrades when the distribution of $\theta$ become ``fully'' complex-valued.

 \begin{figure}[t]
  \centering
  \includegraphics[width=.47\linewidth]{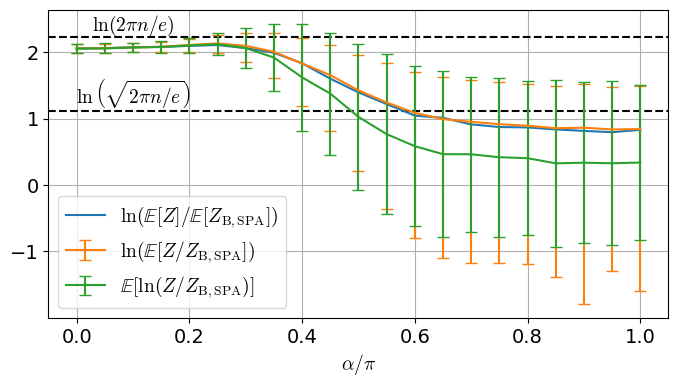}
  \hspace{0.5em}
  \includegraphics[width=.47\linewidth]{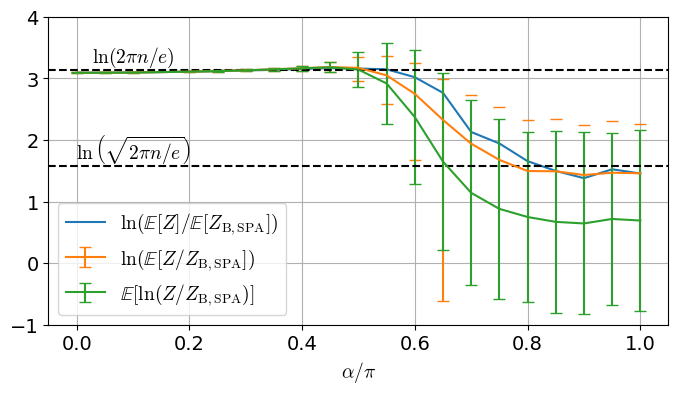}
  \caption{Numerical results of the SPA on DE-NFG $\sfnDE(\btheta)$ for $n=4$ (left) and $n=10$ (right).}
  \label{fig:DENFG:SPA:alpha:supp}
 \end{figure}


 \section{Double Covers of the DE-NFG \(\sfnDE(\btheta)\)}\label{appx:double:cover}
 \subsection{Definition and examples}
 \begin{definition}
  A graph
  $\widetilde{\mathsf{G}} = (\widetilde{\mathcal{V}},
   \widetilde{\mathcal{E}})$ is a cover of a graph
  $\mathsf{G} = (\mathcal{V}, \mathcal{E})$ if there exists a graph
  homomorphism $\pi: \widetilde{\mathcal{V}} \to \mathcal{V}$ such that for each
  $v \in \mathcal{V}$ and $\widetilde{v} \in \pi^{-1}(v)$, the neighborhood of
  $\widetilde{v}$ is mapped bijectively to the neighborhood of $v$.

  Given a cover $\widetilde{\mathsf{G}} = (\widetilde{\mathcal{V}},
  \widetilde{\mathcal{E}})$, if there is a positive integer $M$ such that
  $|\pi^{-1}(v)| = M$ for all $v\in \mathcal{V}$, then $\widetilde{\mathsf{G}}$
  is called an $M$-cover. The set of all $M$-covers of $\mathsf{G}$ is denoted
  by $\widetilde{\mathcal{G}}_M$. \edefn
 \end{definition}


 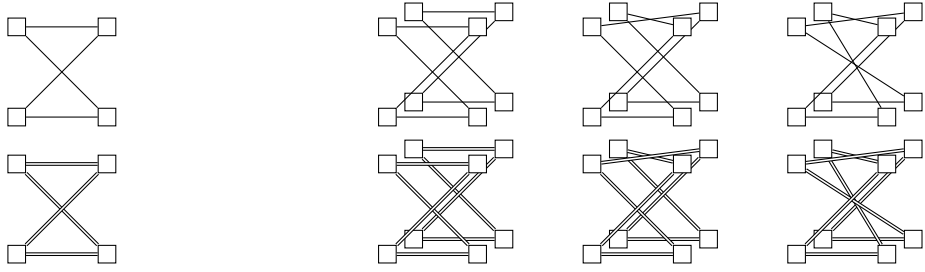
\begin{figure}[t]
  \begin{minipage}[t]{0.45\linewidth}
\centering
\begin{tikzpicture}[node distance=1.2cm, on grid, auto]
\tikzstyle{state}=[shape=rectangle, draw, minimum size=2mm, fill=white]
\begin{pgfonlayer}{main}
\node[state] (f11) [] {};
\node[state, right of=f11] (f12) [] {};
\node[state, below of=f11] (f21) [] {};
\node[state, right of=f21] (f22) [] {};
\end{pgfonlayer}
\begin{pgfonlayer}{main}
\draw[] (f11) -- (f12);
\draw[] (f11) -- (f22);
\draw[] (f21) -- (f12);
\draw[] (f21) -- (f22);
\end{pgfonlayer}
\end{tikzpicture}
\end{minipage}
\begin{minipage}[t]{0.15\linewidth}
\begin{tikzpicture}[node distance=1.2cm, on grid, auto]
\tikzstyle{state}=[shape=rectangle, draw, minimum size=2mm, fill=white]
\begin{pgfonlayer}{main}
\node[state] (f11) [] {};
\node[state,right of=f11] (f12) [] {};
\node[state,below of=f11] (f21) [] {};
\node[state,right of=f21] (f22) [] {};
\end{pgfonlayer}

\begin{pgfonlayer}{above}
\node[state] (g11) [below left=0.2cm and 0.35cm of f11]  {};
\node[state,right of=g11] (g12) [] {};
\node[state,below of=g11] (g21) [] {};
\node[state,right of=g21] (g22) [] {};
\end{pgfonlayer}
\begin{pgfonlayer}{main}
\draw[] (f11) -- (f12);
\draw[] (f11) -- (f22);
\draw[] (f21) -- (f12);
\draw[] (f21) -- (f22);
\end{pgfonlayer}
\begin{pgfonlayer}{above}
\draw[] (g11) -- (g12);
\draw[] (g11) -- (g22);
\draw[] (g21) -- (g12);
\draw[] (g21) -- (g22);
\end{pgfonlayer}  
\end{tikzpicture}
\end{minipage}%
\begin{minipage}[t]{0.15\linewidth}
\begin{tikzpicture}[node distance=1.2cm, on grid, auto]
\tikzstyle{state}=[shape=rectangle, draw, minimum size=2mm, fill=white]
\begin{pgfonlayer}{main}
\node[state] (f11) [] {};
\node[state,right of=f11] (f12) [] {};
\node[state,below of=f11] (f21) [] {};
\node[state,right of=f21] (f22) [] {};
\end{pgfonlayer}

\begin{pgfonlayer}{above}
\node[state] (g11) [below left=0.2cm and 0.35cm of f11]  {};
\node[state,right of=g11] (g12) [] {};
\node[state,below of=g11] (g21) [] {};
\node[state,right of=g21] (g22) [] {};
\end{pgfonlayer}
\begin{pgfonlayer}{main}
\draw[] (f11) -- (g12);
\draw[] (f11) -- (f22);
\draw[] (f21) -- (f12);
\draw[] (f21) -- (f22);
\end{pgfonlayer}
\begin{pgfonlayer}{above}
\draw[] (g11) -- (f12);
\draw[] (g11) -- (g22);
\draw[] (g21) -- (g12);
\draw[] (g21) -- (g22);
\end{pgfonlayer}
\end{tikzpicture}
\end{minipage}%
\begin{minipage}[t]{0.15\linewidth}
\begin{tikzpicture}[node distance=1.2cm, on grid, auto]
\tikzstyle{state}=[shape=rectangle, draw, minimum size=2mm, fill=white]
\begin{pgfonlayer}{main}
\node[state] (f11) [] {};
\node[state,right of=f11] (f12) [] {};
\node[state,below of=f11] (f21) [] {};
\node[state,right of=f21] (f22) [] {};
\end{pgfonlayer}

\begin{pgfonlayer}{above}
\node[state] (g11) [below left=0.2cm and 0.35cm of f11]  {};
\node[state,right of=g11] (g12) [] {};
\node[state,below of=g11] (g21) [] {};
\node[state,right of=g21] (g22) [] {};
\end{pgfonlayer}
\begin{pgfonlayer}{main}
\draw[] (f11) -- (g12);
\draw[] (f21) -- (f12);
\draw[] (f21) -- (f22);
\end{pgfonlayer}
\begin{pgfonlayer}{above}
\draw[] (f11) -- (g22);
\draw[] (g11) -- (f12);
\draw[] (g11) -- (f22);
\draw[] (g21) -- (g12);
\draw[] (g21) -- (g22);
\end{pgfonlayer}
\end{tikzpicture}
\end{minipage}%

  \vspace{.5em}
  \begin{minipage}[t]{0.45\linewidth}
\centering
\begin{tikzpicture}[node distance=1.2cm, on grid, auto]
\tikzstyle{state}=[shape=rectangle, draw, minimum size=2mm, fill=white]
\begin{pgfonlayer}{main}
\node[state] (f11) [] {};
\node[state, right of=f11] (f12) [] {};
\node[state, below of=f11] (f21) [] {};
\node[state, right of=f21] (f22) [] {};
\end{pgfonlayer}
\begin{pgfonlayer}{main}
\draw[double] (f11) -- (f12);
\draw[double] (f11) -- (f22);
\draw[double] (f21) -- (f12);
\draw[double] (f21) -- (f22);
\end{pgfonlayer}
\end{tikzpicture}
\end{minipage}
\begin{minipage}[t]{0.15\linewidth}
\begin{tikzpicture}[node distance=1.2cm, on grid, auto]
\tikzstyle{state}=[shape=rectangle, draw, minimum size=2mm, fill=white]
\begin{pgfonlayer}{main}
\node[state] (f11) [] {};
\node[state,right of=f11] (f12) [] {};
\node[state,below of=f11] (f21) [] {};
\node[state,right of=f21] (f22) [] {};
\end{pgfonlayer}

\begin{pgfonlayer}{above}
\node[state] (g11) [below left=0.2cm and 0.35cm of f11]  {};
\node[state,right of=g11] (g12) [] {};
\node[state,below of=g11] (g21) [] {};
\node[state,right of=g21] (g22) [] {};
\end{pgfonlayer}
\begin{pgfonlayer}{main}
\draw[double] (f11) -- (f12);
\draw[double] (f11) -- (f22);
\draw[double] (f21) -- (f12);
\draw[double] (f21) -- (f22);
\end{pgfonlayer}
\begin{pgfonlayer}{above}
\draw[double] (g11) -- (g12);
\draw[double] (g11) -- (g22);
\draw[double] (g21) -- (g12);
\draw[double] (g21) -- (g22);
\end{pgfonlayer}  
\end{tikzpicture}
\end{minipage}%
\begin{minipage}[t]{0.15\linewidth}
\begin{tikzpicture}[node distance=1.2cm, on grid, auto]
\tikzstyle{state}=[shape=rectangle, draw, minimum size=2mm, fill=white]
\begin{pgfonlayer}{main}
\node[state] (f11) [] {};
\node[state,right of=f11] (f12) [] {};
\node[state,below of=f11] (f21) [] {};
\node[state,right of=f21] (f22) [] {};
\end{pgfonlayer}

\begin{pgfonlayer}{above}
\node[state] (g11) [below left=0.2cm and 0.35cm of f11]  {};
\node[state,right of=g11] (g12) [] {};
\node[state,below of=g11] (g21) [] {};
\node[state,right of=g21] (g22) [] {};
\end{pgfonlayer}
\begin{pgfonlayer}{main}
\draw[double] (f11) -- (g12);
\draw[double] (f11) -- (f22);
\draw[double] (f21) -- (f12);
\draw[double] (f21) -- (f22);
\end{pgfonlayer}
\begin{pgfonlayer}{above}
\draw[double] (g11) -- (f12);
\draw[double] (g11) -- (g22);
\draw[double] (g21) -- (g12);
\draw[double] (g21) -- (g22);
\end{pgfonlayer}
\end{tikzpicture}
\end{minipage}%
\begin{minipage}[t]{0.15\linewidth}
\begin{tikzpicture}[node distance=1.2cm, on grid, auto]
\tikzstyle{state}=[shape=rectangle, draw, minimum size=2mm, fill=white]
\begin{pgfonlayer}{main}
\node[state] (f11) [] {};
\node[state,right of=f11] (f12) [] {};
\node[state,below of=f11] (f21) [] {};
\node[state,right of=f21] (f22) [] {};
\end{pgfonlayer}

\begin{pgfonlayer}{above}
\node[state] (g11) [below left=0.2cm and 0.35cm of f11]  {};
\node[state,right of=g11] (g12) [] {};
\node[state,below of=g11] (g21) [] {};
\node[state,right of=g21] (g22) [] {};
\end{pgfonlayer}
\begin{pgfonlayer}{main}
\draw[double] (f11) -- (g12);
\draw[double] (f21) -- (f12);
\draw[double] (f21) -- (f22);
\end{pgfonlayer}
\begin{pgfonlayer}{above}
\draw[double] (f11) -- (g22);
\draw[double] (g11) -- (f12);
\draw[double] (g11) -- (f22);
\draw[double] (g21) -- (g12);
\draw[double] (g21) -- (g22);
\end{pgfonlayer}
\end{tikzpicture}
\end{minipage}%
  \caption{NFG $\sfn(\btheta)$ (top left) and possible double-covers of it (top right).
   DE-NFG $\sfnDE(\btheta)$ (bottom left) and possible double-covers of it (bottom right).}
  \label{fig:doublecover}
 \end{figure}

 Fig.~\ref{fig:doublecover} shows examples of double-covers of NFG
 $\sfn(\btheta)$ and DE-NFG $\sfnDE(\btheta)$, respectively. We define the
 degree-$2$ Bethe partition sum of $\sfnDE(\btheta)$ to be
 \[
  Z_{\B,2}\bigl(\sfnDE(\btheta)\bigr)\defeq
  \sqrt{\Braket{Z\bigl(\widetilde{\sfn}_{\mathrm{DE}}(\btheta)\bigr)}_{\widetilde{\sfn}_{\mathrm{DE}}\in\widetilde{\mathcal N}_{\mathrm{DE},2}}}.
 \]


 \subsection{Proof of Proposition~\ref{prop:double:cover:decomp}}
 \begin{figure}[t]
  \centering
  \tikzset{
state/.style={shape=rectangle, draw, minimum size=2mm, fill=white, node distance=1.5cm}
}
\begin{minipage}[t]{0.3\linewidth}
\centering
\begin{tikzpicture}[ on grid, auto]\
\begin{pgfonlayer}{behind}
\node[state, label=above:$\tfl{1,2}$] (f11) [] {};
\node[state,right of=f11] (f12) [] {};
\node[state,below of=f11] (f21) [] {};
\node[state,right of=f21] (f22) [] {};
\end{pgfonlayer}
\begin{pgfonlayer}{above}
\node[state, label=above:$\tfl{1,1}$] (g11) [below left=0.4cm and 0.6cm of f11]  {};
\node[state,right of=g11] (g12) [] {};
\node[state,below of=g11] (g21) [] {};
\node[state,right of=g21] (g22) [] {};
\end{pgfonlayer}
\begin{pgfonlayer}{behind}
\draw[double,red] (f11) -- (g12);
\draw[double,blue] (f11) -- (f22);
\draw[double,blue] (f21) -- (f12);
\draw[double,blue] (f21) -- (f22);
\end{pgfonlayer}
\begin{pgfonlayer}{above}
\draw[double,red] (g11) -- (f12);
\draw[double,black] (g11) -- (g22);
\draw[double,black] (g21) -- (g12);
\draw[double,black] (g21) -- (g22);
\end{pgfonlayer}
\end{tikzpicture}
\end{minipage}%
\begin{minipage}[t]{0.35\linewidth}
\centering
\begin{tikzpicture}[on grid, auto]
\begin{pgfonlayer}{behind}
\node[state, label=above:$\fl{1,2}$] (F11) [] {};
\node[state,right of=F11] (F12) [] {};
\node[state,below of=F11] (F21) [] {};
\node[state,right of=F21] (F22) [] {};
\end{pgfonlayer}
\begin{pgfonlayer}{background}
\node[state, label=above:$f^{\mathrm{L}'}_{1,2}$] (f11) [below left=-0.4cm and -0.6cm of F11]  {};
\node[state,right of=f11] (f12) [] {};
\node[state,below of=f11] (f21) [] {};
\node[state,right of=f21] (f22) [] {};
\end{pgfonlayer}
\begin{pgfonlayer}{above}
\node[state,label=above:$f^{\mathrm{L}'}_{1,1}$] (g11) [below left=0.4cm and 0.6cm of F11]  {};
\node[state,right of=g11] (g12) [] {};
\node[state,below of=g11] (g21) [] {};
\node[state,right of=g21] (g22) [] {};
\end{pgfonlayer}
\begin{pgfonlayer}{glass}
\node[state, label=above:$\fl{1,1}$] (G11) [below left=0.4cm and 0.6cm of g11]  {};
\node[state,right of=G11] (G12) [] {};
\node[state,below of=G11] (G21) [] {};
\node[state,right of=G21] (G22) [] {};
\end{pgfonlayer}
\begin{pgfonlayer}{behind}
\draw[red] (F11) -- (G12);
\draw[blue] (F11) -- (F22);
\draw[blue] (F21) -- (F12);
\draw[blue] (F21) -- (F22);
\end{pgfonlayer}
\begin{pgfonlayer}{background}
\draw[red] (f11) -- (g12);
\draw[blue] (f11) -- (f22);
\draw[blue] (f21) -- (f12);
\draw[blue] (f21) -- (f22);
\end{pgfonlayer}
\begin{pgfonlayer}{above}
\draw[red] (g11) -- (f12);
\draw[black] (g11) -- (g22);
\draw[black] (g21) -- (g12);
\draw[black] (g21) -- (g22);
\end{pgfonlayer}
\begin{pgfonlayer}{glass}
\draw[red] (G11) -- (F12);
\draw[black] (G11) -- (G22);
\draw[black] (G21) -- (G12);
\draw[black] (G21) -- (G22);
\end{pgfonlayer}
\end{tikzpicture}
\end{minipage}%
\begin{minipage}[t]{0.35\linewidth}
\centering
\begin{tikzpicture}[on grid, auto]
\begin{pgfonlayer}{above}
\node[state, label=above:$\fl{1,2}$] (F11) [] {};
\node[state,right of=F11] (F12) [] {};
\node[state,below of=F11] (F21) [] {};
\node[state,right of=F21] (F22) [] {};
\end{pgfonlayer}
\begin{pgfonlayer}{background}
\node[state, label=above:$f^{\mathrm{L}'}_{1,2}$] (f11) [below left=-0.8cm and -1.4cm of F11]  {};
\node[state,right of=f11] (f12) [] {};
\node[state,below of=f11] (f21) [] {};
\node[state,right of=f21] (f22) [] {};
\end{pgfonlayer}
\begin{pgfonlayer}{behind}
\node[state,label=above:$f^{\mathrm{L}'}_{1,1}$] (g11) [below left=-.4cm and -.8cm of F11]  {};
\node[state,right of=g11] (g12) [] {};
\node[state,below of=g11] (g21) [] {};
\node[state,right of=g21] (g22) [] {};
\end{pgfonlayer}
\begin{pgfonlayer}{glass}
\node[state, label=above:$\fl{1,1}$] (G11) [below left=0.4cm and 0.6cm of F11]  {};
\node[state,right of=G11] (G12) [] {};
\node[state,below of=G11] (G21) [] {};
\node[state,right of=G21] (G22) [] {};
\end{pgfonlayer}
\begin{pgfonlayer}{above}
\draw[red] (F11) -- (G12);
\draw[blue] (F11) -- (F22);
\draw[blue] (F21) -- (F12);
\draw[blue] (F21) -- (F22);
\end{pgfonlayer}
\begin{pgfonlayer}{background}
\draw[red] (f11) -- (g12);
\draw[blue] (f11) -- (f22);
\draw[blue] (f21) -- (f12);
\draw[blue] (f21) -- (f22);
\end{pgfonlayer}
\begin{pgfonlayer}{behind}
\draw[red] (g11) -- (f12);
\draw[black] (g11) -- (g22);
\draw[black] (g21) -- (g12);
\draw[black] (g21) -- (g22);
\end{pgfonlayer}
\begin{pgfonlayer}{glass}
\draw[red] (G11) -- (F12);
\draw[black] (G11) -- (G22);
\draw[black] (G21) -- (G12);
\draw[black] (G21) -- (G22);
\end{pgfonlayer}
\end{tikzpicture}
\end{minipage}%
  \caption{A possible $2$-cover of $\sfnDE(\btheta)$ ($n=2$) and its decomposition.} \label{fig:doublecover:DENFG:decomp}
 \end{figure}
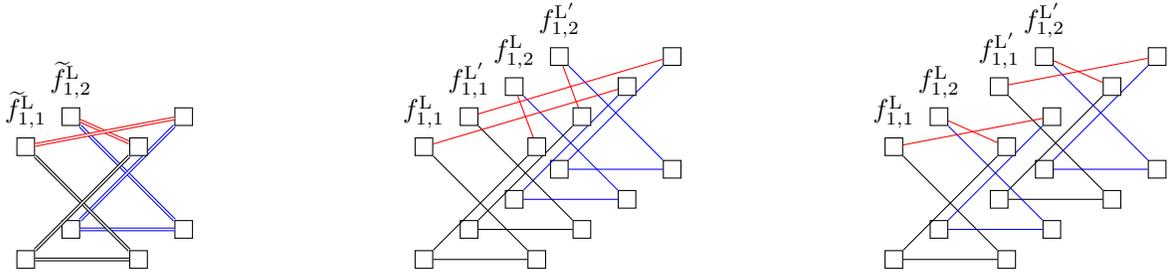

 Fig.~\ref{fig:doublecover:DENFG:decomp} illustrates the idea behind the
 decomposition. Recall the definitions of $\tfl{i}$ and $\mathbf W_{i,j}$ in
 \(\sfnDE(\btheta)\), i.e.,
 \begin{align*}
  \tfl{i}\bigl(\{\txl{i,j}\}_{j\in[n]}\bigr) & =
  \fl{i}\bigl(\{ \xl{i,j} \}_{j \in [n]}\bigr)\cdot
  \fl{i}\bigl(\{ \xlprime{i,j} \}_{j \in [n]}\bigr), \quad
  \txl{i,j} = (\xl{i,j}, \xlprime{i,j}),         \\
  \mathbf W_{i,j}                            & =
  \Biggl[\begin{matrix}
           1            & \overline{\theta_{i,j}} \\
           \theta_{i,j} & |\theta_{i,j}|^2
          \end{matrix}\Biggr] =
  \Biggl[\begin{matrix}
           1 \\ \theta_{i,j}
          \end{matrix}\Biggr]
  \Bigl[\begin{matrix}
          1 & \overline{\theta_{i,j}}
         \end{matrix}\Bigr].
 \end{align*}

 In a $2$-cover $\widetilde{\sfn}_\mathrm{DE}(\btheta)$
 (Fig.~\ref{fig:doublecover:DENFG:decomp} (left)), each function node,
 originally defined over $\mathcal X^{2n}$, can be decomposed into the direct
 product of two functions over $\mathcal X^{n}$. Taking $\tfl{i,1}\otimes
 \tfl{i,2}$ as an example, we have
 \begin{align*}
  \tfl{i,1}\bigl( \{\txl{i,j,1}\}_{j\in[n]} \bigr) & =
  \fl{i,1}\bigl( \{x_{i,j,1}\}_{j\in[n]} \bigr) \cdot
  f^{\mathrm{L}'}_{i,1}\bigl( \{x'_{i,j,1}\}_{j\in[n]} \bigr), \\
  \tfl{i,2}\bigl( \{\txl{i,j,2}\}_{j\in[n]} \bigr) & =
  \fl{i,2}\bigl( \{x_{i,j,2}\}_{j\in[n]} \bigr) \cdot
  f^{\mathrm{L}'}_{i,2}\bigl( \{x'_{i,j,2}\}_{j\in[n]} \bigr),
 \end{align*}
 where $\fl{i,1},f^{\mathrm{L}'}_{i,1}, \fl{i,2}, f^{\mathrm{L}'}_{i,2}$ are the same function as $\fl{i}$.

 Further decomposing each variable \(\txl{i,j,m}\) on a double-edge (with
 alphabet $\mathcal X^2$) into two variables \(x_{i,j,m}\) and \(x'_{i,j,m}\)
 on single-edges (each with alphabet $\mathcal X$), we obtain
 Fig.~\ref{fig:doublecover:DENFG:decomp} (middle). Finally, we re-group the
 umprimed function nodes and variables (e.g., \(\fl{i,1}, \fl{i,2},
 x_{i,j,m}\)) into a separate NFG, and the primed ones (e.g.,
 \(f^{\mathrm{L}'}_{i,1}, f^{\mathrm{L}'}_{i,2}, x'_{i,j,m}\)) into another NFG
 (Fig.~\ref{fig:doublecover:DENFG:decomp} (right)). We refer to the former as
 the ``unprimed NFG'' and the latter as the ``primed NFG''. As a result, the
 umprimed NFG and the primed NFG are identical, except for the edge weights
 being $\theta_{i,j}$ and $\overline{\theta_{i,j}}$, respectively.

 For any double cover $\widetilde{\sfn}_{\mathrm{DE}}(\btheta)$ of $\sfnDE(\btheta)$,
 the above decomposition gives a
 unique $\widetilde{\sfn}\in \widetilde {\mathcal N}_2$ such that
 $\widetilde{\sfn}_{\mathrm{DE}}(\btheta)$is the disjoint union of the unprimed NFG
 $\widetilde {\sfn}(\btheta)$ and the primed NFG $\widetilde {\sfn}(\overline
 {\btheta})$. Hence, the partition sum of
 $\widetilde{\sfn}_{\mathrm{DE}}(\btheta)$ can be expressed as
 \[Z\bigl(\widetilde{\sfn}_{\mathrm{DE}}(\btheta)\bigr)=
  Z\bigl(\widetilde{\sfn}(\btheta)\bigr)\cdot Z\bigl(\widetilde{\sfn}(\overline{\btheta})\bigr)=
  Z\bigl(\widetilde{\sfn}(\btheta)\bigr)\cdot \overline{Z\bigl(\widetilde{\sfn}(\btheta)\bigr)}=
  \big|Z\bigl(\widetilde{\sfn}(\btheta)\bigr)\big|^2.\]
 As a consequence, we have
 \[Z_{\B,2}\bigl(\sfnDE(\btheta)\bigr)
  =\sqrt{\Braket{Z\bigl(\widetilde{\sfn}_{\mathrm{DE}}(\btheta)\bigr)}
   _{\widetilde{\sfn}_{\mathrm{DE}}\in\widetilde{\mathcal N}_{\mathrm{DE},2}}}
  =\sqrt{\Big\langle\big| Z \bigl(\widetilde{\sfn}(\btheta)\bigr)\big|^2\Big\rangle
   _{\widetilde{\sfn}\in \widetilde {\mathcal N}_2}}.\]

 \section{Counting in Double Covers of the DE-NFG \(\sfnDE(\btheta)\)}\label{appx:double:cover:counting}
 \begin{figure}[t]\centering
  \input{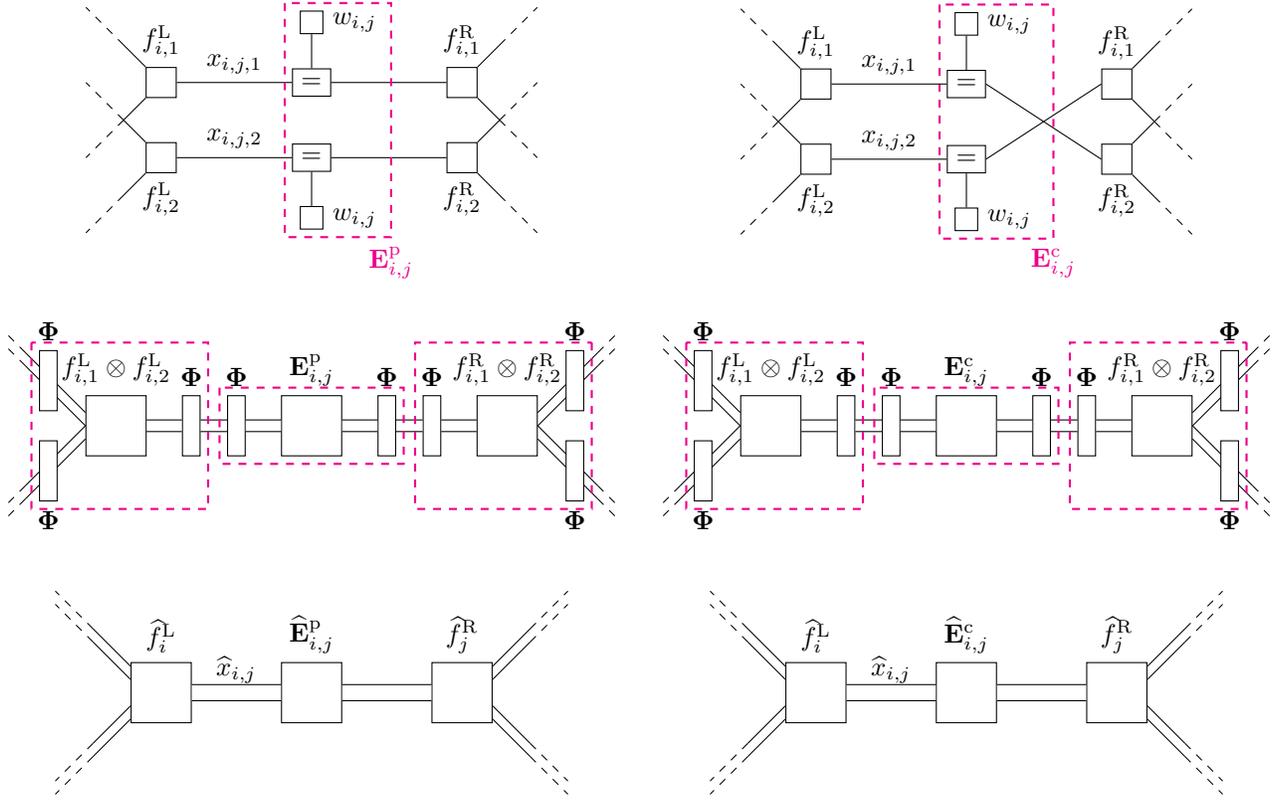}
  \caption{Part of a possible $2$-cover of $\sfn(\btheta)$ and NFG transform by $\boldsymbol{\Phi}$:
   Copies of the edge $(i,j)$ are parallel (left) and crossed (right).}\label{fig:doublecover:NFG:CTB}
 \end{figure}

 Towards evaluating $\E[Z_{\B,2}^2(\sfnDE(\btheta))]$, we use the NFG transform techniques 
 on double covers proposed in \cite{vontobelAnalysisDoubleCovers2016} 
 separately to the unprimed NFG and the primed NFG in Fig.~\ref{fig:doublecover:DENFG:decomp} (right). 
 We only show the transformations applied to the unprimed NFG. The transformations for the primed NFG are
 the same, except for replacing $\theta_{i,j}$ by its complex conjugate,
 $\overline{\theta_{i,j}}$, in the edge weight functions.

 We use ``\underline{p}arallel'' and ``\underline{c}rossed'' to refer to the
 two possible choices of edges in a double cover (see
 Fig.~\ref{fig:doublecover:NFG:CTB} (top)). By closing-the-box in
 Fig.~\ref{fig:doublecover:NFG:CTB} (top), we obtain the matrix representation
 of the edge weight functions, respectively,
 \[
  \mathbf E_{i,j}^{\mathrm{p}}=\begin{bmatrix}
   1 & 0            & 0            & 0               \\
   0 & \theta_{i,j} & 0            & 0               \\
   0 & 0            & \theta_{i,j} & 0               \\
   0 & 0            & 0            & \theta_{i,j} ^2
  \end{bmatrix} \text{ and }
  \mathbf E_{i,j}^{\mathrm{c}}=\begin{bmatrix}
   1 & 0            & 0            & 0               \\
   0 & 0            & \theta_{i,j} & 0               \\
   0 & \theta_{i,j} & 0            & 0               \\
   0 & 0            & 0            & \theta_{i,j} ^2
  \end{bmatrix}.
 \]

 Note that the matrices $\mathbf E_{i,j}^{\mathrm{p}}$ and $\mathbf
 E_{i,j}^{\mathrm{c}}$ commute, and hence, they are simultaneously diagonalized
 by
 \[
  \boldsymbol{\Phi}\defeq\begin{bmatrix}
   1 & 0          & 0           & 0 \\
   0 & 1/\sqrt{2} & 1/\sqrt{2}  & 0 \\
   0 & 1/\sqrt{2} & -1/\sqrt{2} & 0 \\
   0 & 0          & 0           & 1
  \end{bmatrix}.
 \]
 Applying the DE-NFG transformation by $\boldsymbol \Phi$ on
 $\widetilde{\sfn}\in\widetilde{\mathcal{N}}_2$ (see
 Fig.~\ref{fig:doublecover:NFG:CTB} (middle and bottom)), we obtain the
 following for the check node function $\fl{i,1}\otimes \fl {i,2}$:
 \[
  \hatfl{i}\big(\set{\hatx{i,j}}_{j\in[n]}\big)
  =
  \begin{cases}
  1  & \exists k \in [n], \hatx{i,k}=(1,1);
  \ \hatx{i,j}=(0,0),j\neq k                                                \\
  1  & \exists k, \ell \in [n], k\neq \ell, \hatx{i,k}=\hatx{i,\ell}=(0,1);
  \ \hatx{i,j}=(0,0),j\neq k,\ell                                           \\
  -1 & \exists k, \ell \in [n], k\neq \ell, \hatx{i,k}=\hatx{i,\ell}=(1,0);
  \ \hatx{i,j}=(0,0),j\neq k,\ell                                           \\
  0  & \text{otherwise}
  \end{cases}.
 \]
 The function $\hatfr{j}$ is defined similarly and hence omitted. For edge
 weights $\mathbf E^{\mathrm p}_{i,j}$ and $\mathbf E^{\mathrm c}_{i,j}$, the matrix
 representation of the transformed edge weight functions are given by,
 respectively,
 \[
  \widehat{\mathbf E}^{\mathrm p}_{i,j}\defeq \boldsymbol{\Phi} \mathbf{E}^{\mathrm p}_{i,j} \boldsymbol{\Phi} =
  \begin{bmatrix}
   1 & 0            & 0            & 0               \\
   0 & \theta_{i,j} & 0            & 0               \\
   0 & 0            & \theta_{i,j} & 0               \\
   0 & 0            & 0            & \theta_{i,j} ^2
  \end{bmatrix}\text{ and }
  \widehat{\mathbf E}^{\mathrm c}_{i,j}\defeq \boldsymbol{\Phi} \mathbf{E}^{\mathrm c}_{i,j} \boldsymbol{\Phi} =
  \begin{bmatrix}
   1 & 0            & 0             & 0               \\
   0 & \theta_{i,j} & 0             & 0               \\
   0 & 0            & -\theta_{i,j} & 0               \\
   0 & 0            & 0             & \theta_{i,j} ^2
  \end{bmatrix}.
 \]
 Denote the transformed NFG of $\widetilde{\sfn}(\btheta)$ by
 $\widehat{\sfn}(\btheta)$, and the set of all $2$-covers $\widetilde{\mathcal
 N}_2$ is mapped to $\widehat{\mathcal N}_2$. Since elements of
 $\widehat{\mathcal N}_2$ only differ by the edge weight function, the set of
 valid configurations is independent of $\widehat{\sfn}\in\widehat{\mathcal
 N}_2$, so we denote the set of valid configurations by $\validc$. By the
 definition of $\hatfl{i}$ and $\hatfr{j}$, we see that for any valid
 configuration in $\validc$, every vertex of $\widehat{\sfn}(\btheta)$ is one
 of the three cases:
 \begin{itemize}
  \item the endpoint of exactly one $(1,1)$-edge,
  \item vertex in exactly one $(0,1)$-cycle,
  \item vertex in exactly one $(1,0)$-cycle.
 \end{itemize}

 Combining Proposition \ref{prop:double:cover:decomp}, Appendix
 \ref{appx:double:cover} and the results so far in Appendix
 \ref{appx:double:cover:counting}, we have
 \begin{align*}
  Z_{\B,2}^2(\sfnDE(\btheta))
   & = \Braket{Z\bigl(\widetilde{\sfn}_{\mathrm{DE}}(\btheta)\bigr)}_{\widetilde{\sfn}_{\mathrm{DE}}\in\widetilde{\mathcal N}_{\mathrm{DE},2}} \\
   & = \Big\langle\big| Z \bigl(\widetilde{\sfn}(\btheta)\bigr)\big|^2\Big\rangle_{\widetilde{\sfn}\in \widetilde {\mathcal N}_2}              \\
   & = \Big\langle\big| Z \bigl(\widehat{\sfn}(\btheta)\bigr)\big|^2\Big\rangle_{\widehat{\sfn}\in \widehat {\mathcal N}_2}                    \\
   & = 2^{-n^2}\sum_{\widehat{\sfn}\in\widehat{\mathcal N}_2}\Biggl\vert
  \sum_{\hatbfx \in \validc}
  g_{\widehat{\sfn}}(\hatbfx)\Biggr\vert^2                                                                                                     \\
   & =2^{-n^2}
  \sum_{\substack{\hatbfx[(1)]\in \validc                                                                                                      \\\hatbfx[(2)]\in \validc}} \,
  \sum_{\widehat{\sfn}\in\widehat{\mathcal N}_2}
  g_{\widehat{\sfn}}(\hatbfx[(1)])\overline{g_{\widehat{\sfn}}(\hatbfx[(2)])}.
 \end{align*}


 \section{Proof of Theorem \ref{thm:double:cover:allone}}\label{appx:double:cover:allone}
 Recall that we want to evaluate the following quantity for $\btheta=\mathbf 1_n$:
 \[
  Z_{\B,2}^2(\sfnDE(\btheta))=2^{-n^2}
  \sum_{\substack{\hatbfx[(1)]\in \validc\\\hatbfx[(2)]\in \validc}} \,
  \sum_{\widehat{\sfn}\in\widehat{\mathcal N}_2}
  g_{\widehat{\sfn}}(\hatbfx[(1)])g_{\widehat{\sfn}}(\hatbfx[(2)])\]
 For $(\hatbfx[(1)],\hatbfx[(2)])\in \validc\times \validc$, define
 \[
  g(\hatbfx[(1)],\hatbfx[(2)])\defeq
  2^{-n^2}\sum_{\widehat{\sfn}\in\widehat{\mathcal N}_2}
  g_{\widehat{\sfn}}(\hatbfx[(1)])g_{\widehat{\sfn}}(\hatbfx[(2)]).
 \]

 \begin{lem}
  For $(\hatbfx[(1)],\hatbfx[(2)])\in \validc\times \validc$, if there exists $(i,j)\in[n]^2$
  such that $\hatx{i,j}^{(1)}=(1,0)$ but $\hatx{i,j}^{(2)}\neq (1,0)$, then $g(\hatbfx[(1)],\hatbfx[(2)])=0$.
 \end{lem}

 \begin{proof}
  Fix such an index $(i,j)$. Partition $\widehat{\mathcal N}_2$ into two subsets according to
  the $(i,j)$-th edge weight:
  \[
   \widehat{\mathcal N}_2
   =\widehat{\mathcal N}_2^{i,j,\mathrm p}\,\sqcup\,\widehat{\mathcal N}_2^{i,j,\mathrm c},
  \]
  where $\widehat{\mathcal N}_2^{i,j,\mathrm p}$ consists of those
  $\widehat{\sfn}$ whose $(i,j)$-th edge has weight matrix $\widehat{\mathbf
  E}^{\mathrm p}_{i,j}$, and $\widehat{\mathcal N}_2^{i,j,\mathrm c}$ consists
  of those $\widehat{\sfn}$ whose $(i,j)$-th edge has weight matrix
  $\widehat{\mathbf E}^{\mathrm c}_{i,j}$.

  Define a pairing map
  \[
   \kappa:\widehat{\mathcal N}_2^{i,j,\mathrm p}\to \widehat{\mathcal N}_2^{i,j,\mathrm c}
  \]
  by letting $\kappa(\widehat{\sfn})=\widehat{\sfn}'$ be the factor graph
  obtained from $\widehat{\sfn}$ by changing only the $(i,j)$-th edge weight
  matrix from $\widehat{\mathbf E}^{\mathrm p}_{i,j}$ to $\widehat{\mathbf
  E}^{\mathrm c}_{i,j}$ (and leaving all other edges unchanged). Then $\kappa$
  is a bijection, so $\bigl|\widehat{\mathcal N}_2^{i,j,\mathrm c}\bigr|
  =\bigl|\widehat{\mathcal N}_2^{i,j,\mathrm p}\bigr|$ and $\widehat{\mathcal
  N}_2$ is partitioned into disjoint pairs
  $\{\widehat{\sfn},\kappa(\widehat{\sfn})\}$.

  By the assumption that $\hatx{i,j}^{(1)}=(1,0)$, changing the $(i,j)$-th edge
  weight changes the sign of $g_{\widehat{\sfn}}(\hatbfx[(1)])$:
  \[
   g_{\kappa(\widehat{\sfn})}(\hatbfx[(1)])=-\,g_{\widehat{\sfn}}(\hatbfx[(1)]).
  \]
  On the other hand, since $\hatx{i,j}^{(2)}\neq (1,0)$, changing the
  $(i,j)$-th edge weight does not change $g_{\widehat{\sfn}}(\hatbfx[(2)])$:
  \[
   g_{\kappa(\widehat{\sfn})}(\hatbfx[(2)])=g_{\widehat{\sfn}}(\hatbfx[(2)]).
  \]
  Therefore, for every $\widehat{\sfn}\in\widehat{\mathcal N}_2^{i,j,\mathrm
  c}$,
  \[
   g_{\widehat{\sfn}}(\hatbfx[(1)])g_{\widehat{\sfn}}(\hatbfx[(2)])
   +
   g_{\kappa(\widehat{\sfn})}(\hatbfx[(1)])g_{\kappa(\widehat{\sfn})}(\hatbfx[(2)])=0.
  \]
  Summing over all pairs yields cancellation:
  \[
   g(\hatbfx[(1)],\hatbfx[(2)])
   =2^{-n^2}\sum_{\widehat{\sfn}\in\widehat{\mathcal N}_2}
   g_{\widehat{\sfn}}(\hatbfx[(1)])g_{\widehat{\sfn}}(\hatbfx[(2)])=0. \qedhere
  \]
 \end{proof}

 Hence $g(\hatbfx[(1)],\hatbfx[(2)])\neq 0$ if and only if $\hatbfx[(1)]$ and
 $\hatbfx[(2)]$ have the same $(1,0)$-cycles. When this condition holds, we
 have $g(\hatbfx[(1)],\hatbfx[(2)])=1$, so it remains to count the number of
 such pairs $(\hatbfx[(1)],\hatbfx[(2)])$.

 Recall that for a valid configuration $\hatbfx\in \validc$, it corresponds to
 $2^c$ pairs of permutations in $\mathcal{S}_n$, where $c$ is the number of
 cycles in $\hatbfx$ (see \cite{ngDoublecoverbasedAnalysisBethe2022}). This
 fact will be used extensively in the following discussions. We first choose a
 $(k,k)$-complete bipartite subgraph for $k\in\set{0,1,\dots,n}$. There are
 $\binom{n}{k}^2$ many such subgraphs. On the $(k,k)$-subgraph, we assign only
 $(1,0)$-cycles, and the assignment is the same for $\hatbfx[(1)]$ and
 $\hatbfx[(2)]$. The number of different assignments is
 \[
  \sum_{\substack{\sigma,\tau\in \mathcal{S}_k\\\sigma\tau^{-1} \text{ is a derangement}}}2^{-c(\sigma\tau^{-1})}
  =\sum_{\sigma,\tau\in \mathcal{S}_k}0^{c_1(\sigma\tau^{-1})}2^{-c(\sigma\tau^{-1})}
  =k!\sum_{\sigma\in \mathcal{S}_k}0^{c_1(\sigma)}2^{-c(\sigma)}
  =\Psi_k\Big(0,1,\frac{1}{2}\Big).
 \]
 On the remaining subgraph, we assign $(1,1)$-edges and $(0,1)$-cycles
 independently for $\hatbfx[(1)]$ and $\hatbfx[(2)]$. The number of different
 assignments is
 \[
  \Big(\sum_{\sigma,\tau\in \mathcal{S}_{n-k}}2^{-c(\sigma\tau^{-1})}\Big)^{\! 2}
  =\Big((n-k)!\sum_{\sigma\in \mathcal{S}_{n-k}}2^{-c(\sigma)}\Big)^{\! 2}
  =\Psi^2_{n-k}\Big(1,1,\frac{1}{2}\Big).
 \]
 Putting everything together, we have
 \[
  Z_{\B,2}^2(\sfnDE(\btheta)) = \sum_{k=0}^{n} \binom{n}{k}^{\!\! 2} \cdot \Psi_{k}\Big(0,1,\frac{1}{2}\Big)  \cdot  \Psi_{n-k}^2\Big(1,1,\frac{1}{2}\Big).
 \]

 \section{Proof of Corollary \ref{cor:double:cover:allone}}\label{appx:double:cover:allone:asym}
 Recall that for $\btheta=\mathbf 1_n$, we have
 \[
  Z^2(\sfnDE(\btheta))=(n!)^4
 \]
 and
 \[
  Z_{\B,2}^2(\sfnDE(\btheta)) = \sum_{k=0}^{n} \binom{n}{k}^{\!\! 2}  \cdot  \Psi_{n-k}^2(1,1,1/2) \cdot \Psi_{k}(0,1,1/2).
 \]
 For simplicity, we just write $Z$ and $Z_{\B,2}$, respectively. Define
 \[
  a_k\defeq\Big(\frac{k!}{n!}\Big)^{\! 2}, \quad
  b_k\defeq\Big(\frac{\Psi_{k}(1,1,1/2)}{(k!^2)}\Big)^{\! 2}, \quad
  c_k\defeq\frac{\Psi_{k}(0,1,1/2)}{(k!)^2}, \quad
  k\in[n].
 \] Then we have
 \begin{align*}
  \frac{Z_{\mathrm B,2}^2}{(n!)^4} &
  =\sum_{k=0}^n \Big(\frac{(n-k)!}{n!}\Big)^{\! 2}
  \Big(\frac{\Psi_{n-k}(1,1,1/2)}{((n-k)!)^2}\Big)^{\! 2}
  \frac{\Psi_{k}(0,1,1/2)}{(k!)^2}                                          \\
                                   & = \sum_{k=0}^n a_{n-k}b_{n-k}c_k       \\
                                   & = b_n+\sum_{k=2}^{n}a_{n-k}b_{n-k}c_k,
 \end{align*}
 where the last equality holds due to $a_n=c_0=1,c_1=0$.

 We use the following two lemmas.

 \begin{lem}[See \protect{\cite[Chap.~VI]{flajoletAnalyticCombinatorics2009}}]

  For $a,b\in \R, m\in\R\setminus\Z_{\leq0}$, it holds that
  \[\Psi_n(a,b,m)\sim
   (n!)^2\frac{b^ne^{a/b-m}n^{m-1}}{\Gamma(m)}(1+\mathcal O(n^{-1})),\]
  where $\mathcal O(\,\cdot\,)$ is the usual big-O notation. In particular, we
  have $b_k\sim\frac{e}{\pi k}(1+\mathcal O(k^{-1}))$ and
  $c_k\sim\frac{1}{\sqrt{\pi ke}}(1+\mathcal O(k^{-1}))$.
 \end{lem}
 \begin{proof}
  It is known that the ordinary generating function of the cycle index $Z_n$ is given by
  \[
   G(x)\defeq\sum_{n\geq 0}Z_nx^n
   =\exp\biggl(\sum_{k=1}^\infty \frac{z_kx^k}{k}\biggr).
  \]
  Setting $z_1=a,z_k=mb^k$ for $k\geq 2$, we obtain
  \[
   \sum_{k=1}^\infty \frac{z_kx^k}{k} = (a-mb)x +\sum_{k=1}^\infty \frac{mb^kx^k}{k} = (a-mb)x-m\ln(1-bx).
  \]
  Hence, we have \[\Psi_n(a,b,m)=(n!)^2[x^n]\frac{\exp((a-mb)x)}{(1-bx)^m},\]
  where $[x^n]f(x)$ is the coefficient of $x^n$ in the power series expansion
  of $f(x)$. By \cite[Thm.
  \RomanNumeralCaps{6}.1]{flajoletAnalyticCombinatorics2009}, we have
  \[[x^n]\frac{1}{(1-x)^{m}}\sim \frac{n^{m-1}}{\Gamma(m)}(1+\mathcal{O}(n^{-1})).\]
  Then by \cite[Sec. \RomanNumeralCaps{6}.3 \&
  4]{flajoletAnalyticCombinatorics2009}, we have
  \[
   \Psi_n(a,b,m)\sim(n!)^2\frac{b^ne^{a/b-m}n^{m-1}}{\Gamma(m)}(1+\mathcal O(n^{-1})). \qedhere
  \]
 \end{proof}

 \begin{lem}
  For the quantities defined above, it holds that $\frac{Z_{\mathrm B,2}^2}{(n!)^4}= b_n(1+\mathcal O(n^{-4}))$.
 \end{lem}
 \begin{proof}
  We prove $\sum_{k=2}^{n}a_{n-k}b_{n-k}c_k/b_n=\mathcal O(n^{-4})$ by splitting the summation
  into $2\leq k\leq n/2$ and $n/2 <k\leq n$ parts separately.

  For $c_n$, we use a constant bound for all $n$. Since
  $c_n\sim\frac{1}{\sqrt{\pi ne}}(1+\mathcal O(n^{-1}))$ and is finite for
  every $n$, there exists a constant $K_c$ such that $c_n\leq K_c,\forall n$.

  Part 1: $2\leq k\leq n/2$. In this part, $a_{n-k}$ is dominated by a
  geometric series, that is,
  \[a_{n-k}=\frac{1}{(n(n-1)\cdots(n-k+1))^2}\leq \Big(\frac{2}{n}\Big)^{\! 2k}
   =\Bigl(\frac{4}{n^2}\Bigr)^{\! k}.\]
  For $b_{n-k}/b_n$, since $b_n\sim\frac{e}{\pi n}(1+\mathcal O(n^{-1}))$, we
  have $b^{-1}_n=\frac{n\pi}{e}(1+\mathcal O(n^{-1}))$. Hence, we have
  \[
   \frac{b_{n-k}}{b_n}
   =\frac{n}{n-k}(1+\mathcal O((n-k)^{-1})+\mathcal O (n^{-1}))
   =\frac{n}{n-k}(1+\mathcal O (n^{-1}))
   =\mathcal O(1),\quad 2\leq k\leq n/2.
  \]
  Multiplying together, we have
  \[a_{n-k}b_{n-k}c_k/b_n \leq KK_c \Bigl(\frac{4}{n^2}\Bigr)^{\! k}\]
  for some constant $K$. Therefore
  \[
   \sum_{k=2}^{n/2}a_{n-k}b_{n-k}c_k/b_n
   \leq KK_c\sum_{k=2}^\infty \Bigl(\frac{4}{n^2}\Bigr)^{\! k}
   = KK_c\frac{16/n^4}{1-4/n^2}.
  \]

  Part 2: $n/2 <k\leq n$. In this part, $a_{n-k}$ is super-polynomially small.
  Since
  \[
   \frac{n!}{(n-k)!}\geq\frac{n!}{\lfloor n/2 \rfloor!}= n(n-1)\cdots(\lfloor n/2\rfloor+1)\geq (n/2)^{n/2},
  \]
  as $n-k\leq\lfloor n/2 \rfloor$, we have $a_{n-k}\leq(2/n)^n$.

  In this case, it suffices to bound $b_{n-k}$ also by a constant $K_b$.
  Multiplying together we have
  \[a_{n-k}b_{n-k}c_k/b_n \leq K_bK_cK'n(2/n)^n.\]

  Adding the two parts together, we have
  \begin{align*}
   \sum_{k=2}^{n}a_{n-k}b_{n-k}c_k/b_n
    & = \sum_{2\leq k\leq n/2}a_{n-k}b_{n-k}c_k/b_n+\sum_{n/2< k\leq n}a_{n-k}b_{n-k}c_k/b_n            \\
    & \leq  KK_c\sum_{k=2}^\infty \Big(\frac{4}{n^2}\Big)^{\! k} + \sum_{n/2< k\leq n} K_bK_cK'n(2/n)^n \\
    & =KK_c\frac{16/n^4}{1-4/n^2}+K_bK_cK'n^2(2/n)^n                                                    \\
    & = \mathcal O(n^{-4}). \qedhere
  \end{align*}
 \end{proof}
 With the above two lemmas, we have $\frac{Z_{\mathrm B,2}}{Z}\sim (n!)^{-2}\Psi_n(1,1,1/2)(1+\mathcal O(n^{-4}))=\sqrt{\frac{e}{\pi n}}(1+\mathcal O(n^{-1}))$, which is the promised result.


 \section{Proof of Theorem \ref{thm:double:cover:complex}}\label{appx:double:cover:complex}
 At the core of evaluating \(\E[Z^2]\) and \(\E[Z^2_{\B,2}]\)
 is the following mixed moment of the entries of \(\boldsymbol{\theta}\):
 \[
  \E\Big[\prod_i\theta_{i,\sigma_1(i)}\theta_{i,\sigma_2(i)}\overline{\theta_{i,\tau_1(i)}\theta_{i,\tau_2(i)}}\Big]
  =\prod_i \E\Big[\theta_{i,\sigma_1(i)}\theta_{i,\sigma_2(i)}\overline{\theta_{i,\tau_1(i)}\theta_{i,\tau_2(i)}}\Big]
 \]
 where \(\sigma_1,\sigma_2,\tau_1,\tau_2\in\mathcal{S}_n\). The expectation
 factorizes since the entries of \(\boldsymbol{\theta}\) are independent. If
 there exist some \(i\in[n]\) such that there is a unique column index among
 \(\set{\sigma_1(i),\sigma_2(i),\tau_1(i),\tau_2(i)}\), then the expectation is
 zero due to the zero-mean assumption.

 Fixing valid configurations $\hatbfx[(1)]$ for $\widehat{\sfn}(\btheta)$ and $\hatbfx[(2)]$ 
 for $\widehat{\sfn}(\overline{\btheta})$, we classify when this pair of valid configurations give nonzero
 contribution to \(\E[Z^2(\sfnDE(\btheta))]\) and \(\E[Z^2_{\B,2}(\sfnDE(\btheta))]\).
 In the case $\mu_{2,0}=0$, a non-vanishing contribution in \(\E[Z^2]\)
 requires the following conditions to hold on \(\hatbfx[(1)]\) and
 \(\hatbfx[(2)]\):
 \begin{enumerate}[label=(\arabic*)]
  \item There exists a \((k,k)\)-bipartite complete subgraph where
  \((1,1)\)-edges are assigned in the same way for \(\hatbfx[(1)]\)
  and \(\hatbfx[(2)]\),
  \item \(\hatbfx[(1)]\) and \(\hatbfx[(2)]\) have the same cycles,
  and \((0,1)\) or \((1,0)\) can be assigned in \(\hatbfx[(1)]\) and \(\hatbfx[(2)]\) independently.
 \end{enumerate}
 For evaluating \(\E[Z^2_{\B,2}]\), condition (2) is further strengthened to the following:
 \begin{enumerate}[label=(\arabic*), start=3]
  \item \(\hatbfx[(1)]\) and \(\hatbfx[(2)]\) have the same cycles.
  For each cycle, if it is assigned as \((0,1)\)-cycle (resp. \((1,0)\)-cycle) in \(\hatbfx[(1)]\),
  then it must be assigned as \((0,1)\)-cycle (resp. \((1,0)\)-cycle) in \(\hatbfx[(2)]\), and vice versa.
 \end{enumerate}
 In the case \(\mu_{2,0}\neq 0\), condition (2) and (3) are unchanged, while condition (1) is weakened to:
 \begin{enumerate}[label=(\arabic*), start=4]
  \item There exists a \((k,k)\)-bipartite complete subgraph where
  \((1,1)\)-edges are assigned independently for \(\hatbfx[(1)]\)
  and \(\hatbfx[(2)]\).
 \end{enumerate}

 In the following, we first build the connection between \(\mathcal S_n\times
 \mathcal S_n\) and \(\validc\), and then detail each of the above cases. A
 cycle of a permutation is called as a nontrivial cycle if it has length at
 least two.

 \subsection{Case $\mu_{2,0}=0$: Evaluation of $\E [Z^2]$}
 Define the map
 \[
  h:\Bigl\{(\sigma_1,\sigma_2,\mathbf b)
  \Bigm\vert \sigma_1,\sigma_2\in\mathcal{S}_n,\mathbf b\in \mathcal X^{c(\sigma_1\sigma_2^{-1})}\Bigr\}
  \to \validc,
  \qquad
  (\sigma_1,\sigma_2,\mathbf b)\mapsto \hatbfx=(\hatx{i,j}),
 \]
 where
 \[
  \hatx{i,j}=
  \begin{cases}
  (1,1), & j=\sigma_1(i)=\sigma_2(i),                                         \\
  (0,1), & i\text{ in the $k$-th nontrivial cycle of }\sigma_1\sigma_2^{-1},\
  j\in\{\sigma_1(i),\sigma_2(i)\},\ b_k=0,                                    \\
  (1,0), & i\text{ in the $k$-th nontrivial cycle of }\sigma_1\sigma_2^{-1},\
  j\in\{\sigma_1(i),\sigma_2(i)\},\ b_k=1,                                    \\
  (0,0), & \text{otherwise},
  \end{cases}
 \]
 and the nontrivial cycles of \(\sigma_1\sigma_2^{-1}\) are sorted according to
 the minimal element in each nontrivial cycle. Thus fixed points of
 $\sigma_1\sigma_2^{-1}$ produce $(1,1)$-edges, while each nontrivial cycle
 contributes two possible assignments, encoded by the binary choice $b_k$.

 Define $h(\sigma_1,\sigma_2)\defeq\set{h(\sigma_1,\sigma_2,\mathbf b):\mathbf
 b\in \mathcal X^{c(\sigma_1\sigma_2^{-1})}}\subseteq \mathcal
 C(\widehat{\mathcal N}_2)$, and
 \[
  g_\mathrm{avg} (\sigma_1,\sigma_2)\defeq
  \frac{1}{|h(\sigma_1,\sigma_2)|}\sum_{\hatbfx\in h(\sigma_1,\sigma_2)} g(\hatbfx)
  =\prod_{i}\theta_{i,\sigma_1(i)}\theta_{i,\sigma_2(i)},
 \]
 where $g$ is the global function of the transformed trivial double-cover,
 i.e., $\widehat{\sfn}(\btheta)$ with all edge weights being $\widehat{\mathbf
 E}^{\mathrm p}_{i,j}$. Note that all valid configurations in
 $h(\sigma_1,\sigma_2)$ have the same cycle structure, and $g_{\mathrm{avg}}$
 is averaging over all possible $(0,1)$- or $(1,0)$-cycle assignments. Then we
 have
 \begin{align*}
  \mathrm{perm}^2(\boldsymbol \theta) & =
  \sum_{\sigma_1,\sigma_2\in \mathcal{S}_n}\prod_{i}\theta_{i,\sigma_1(i)}\theta_{i,\sigma_2(i)} =
  \sum_{\sigma_1,\sigma_2\in \mathcal{S}_n} g_\mathrm{avg}(\sigma_1,\sigma_2), \\
  Z^2                                 & =|\mathrm{perm}(\boldsymbol\theta)|^4=
  \Biggl\lvert\sum_{\sigma_1,\sigma_2\in \mathcal{S}_n} g_\mathrm{avg}(\sigma_1,\sigma_2)\Biggr\rvert^2.
 \end{align*}

 Since $\mu_{1,0}=\mu_{2,0}=0$, we have
 $\E\Big[g_\mathrm{avg}(\sigma_1,\sigma_2)\overline{g_\mathrm{avg}(\tau_1,\tau_2)}\Big]\neq
 0$ if and only if $h(\sigma_1,\sigma_2)=h(\tau_1,\tau_2)$. Otherwise, there
 exists some $i\in[n]$ such that
 $\E\Big[\theta_{i,\sigma_1(i)}\theta_{i,\sigma_2(i)}\overline{\theta_{i,\tau_1(i)}\theta_{i,\tau_2(i)}}\Big]=0$.
 When $h(\sigma_1,\sigma_2)=h(\tau_1,\tau_2)$, the expectation depends only on
 the cycle structure of $\sigma_1\sigma_2^{-1}$ and factorizes over its cycles:
 \[
  \E\left[
   g_{\mathrm{avg}}(\sigma_1,\sigma_2)\,
   \overline{g_{\mathrm{avg}}(\tau_1,\tau_2)}
   \right]
  =
  \mu_{2,2}^{\,c_1(\sigma_1\sigma_2^{-1})}
  \prod_{\ell=2}^n (\mu_{1,1}^2)^{\ell c_\ell(\sigma_1\sigma_2^{-1})}.
 \]

 Fixing $(\sigma_1,\sigma_2)$, there are exactly $2^{c(\sigma_1\sigma_2^{-1})}$
 pairs $(\tau_1,\tau_2)$ such that $h(\sigma_1,\sigma_2)=h(\tau_1,\tau_2)$, all
 contributing the same value. Therefore,
 \[
  \E[Z^2]
  =  \sum_{\sigma_1,\sigma_2\in \mathcal{S}_n}
  2^{c(\sigma_1\sigma_2^{-1})}
  \mu_{2,2}^{\,c_1(\sigma_1\sigma_2^{-1})}
  \prod_{\ell=2}^n
  \bigl(\mu_{1,1}^2\bigr)^{\ell c_\ell(\sigma_1\sigma_2^{-1})}
  =\Psi_n(\mu_{2,2},\mu_{1,1}^2,2).
 \]

 \subsection{Case $\mu_{2,0}=0$: Evaluation of $\E [Z_{\B,2}^2]$}
 Now we turn to $Z^2_{\B,2}=\sum_{\sigma_1,\sigma_2,\tau_1,\tau_2}
  g_\mathrm{avg}^{(1)}(\sigma_1,\sigma_2,\tau_1,\tau_2)$ where
 \begin{align*}
  g_\mathrm{avg}^{(1)}(\sigma_1,\sigma_2,\tau_1,\tau_2)   & \defeq
  \frac{1}{|h(\sigma_1,\sigma_2)||h(\tau_1,\tau_2)|}\sum_{\substack{\hatbfx[(1)]\in h(\sigma_1,\sigma_2)                               \\\hatbfx[(2)]\in h(\tau_1,\tau_2)}}
  g_{\mathrm {avg}}^{(2)}(\hatbfx[(1)],\hatbfx[(2)]),                                                                                  \\
  g_{\mathrm {avg}}^{(2)}(\hatbfx[(1)],\hatbfx[(2)])      & \defeq \prod_{ij} g_{i,j,\mathrm{avg}}(\hatx{i,j}^{(1)},\hatx{i,j}^{(2)}), \\
  g_{i,j,\mathrm{avg}}(\hatx{i,j}^{(1)},\hatx{i,j}^{(2)}) & \defeq
  \frac{1}{2}\left[\widehat{\mathbf E}_{i,j}^\mathrm{p} (\hatx{i,j}^{(1)},\hatx{i,j}^{(1)})
   \overline{\widehat{\mathbf E}_{i,j}^\mathrm{p}} (\hatx{i,j}^{(2)},\hatx{i,j}^{(2)})
   +\widehat{\mathbf E}_{i,j}^\mathrm{c} (\hatx{i,j}^{(1)},\hatx{i,j}^{(1)})
   \overline{\widehat{\mathbf E}_{i,j}^\mathrm{c}} (\hatx{i,j}^{(2)},\hatx{i,j}^{(2)})\right].
 \end{align*}
 Note that $g_{\mathrm{avg}}^{(1)}$ is averaging over $h(\sigma_1,\sigma_2)\times h(\tau_1,\tau_2)\subseteq \validc \times \validc$,
 and $g_{\mathrm {avg}}^{(2)}$ is averaging over $\widehat{\sfn}(\btheta)\sqcup\widehat{\sfn}(\overline{\btheta})$.

 Same as above, we have $\E[
 g_\mathrm{avg}^{(1)}(\sigma_1,\sigma_2,\tau_1,\tau_2)]\neq 0$ if and only if
 $h(\sigma_1,\sigma_2)=h(\tau_1,\tau_2)$. We further split the summation in $
 g_\mathrm{avg}^{(1)}(\sigma_1,\sigma_2,\tau_1,\tau_2)$ according to cycles of
 $\sigma_1\sigma_2^{-1}$:
 \begin{itemize}
  \item For $i=\sigma_1\sigma_2^{-1}(i)$, $\prod_j
  g_{i,j,\mathrm{avg}}(\hatx{i,j}^{(1)},\hatx{i,j}^{(2)})=|\theta_{i,\sigma_1(i)}|^4$
  independent of $\hatbfx[(1)],\hatbfx[(2)]$.
  \item For a nontrivial cycle $C$ of $\sigma_1\sigma_2^{-1}$, we have
  $\prod_{ij:i \in C}g_{i,j,\mathrm{avg}}(\hatx{i,j}^{(1)},\hatx{i,j}^{(2)})=
  \prod_{i:i\in C}|\theta_{i,\sigma_1(i)}\theta_{i,\sigma_2(i)}|^2$ if and only
  if $\hatbfx[(1)]$ and $\hatbfx[(2)]$ have the same assignment for
  $\hatx{i,j}$ for all $i\in C$. In other words, if a nontrivial cycle $C$ is
  assigned as a $(0,1)$-cycle in $\widehat{\sfn}(\btheta)$ and as a
  $(1,0)$-cycle in $\widehat{\sfn}(\overline\btheta)$ (or vice versa), then
  $g_{i,j,\mathrm{avg}}(\,\cdot\,,\,\cdot\,)=0$ for $i\in C$.
 \end{itemize}

 Hence, $\E[g_{\mathrm {avg}}^{(2)}(\hatbfx[(1)],\hatbfx[(2)])]\neq 0$ if and
 only if $\hatbfx[(1)]$ and $\hatbfx[(2)]$ have the same assignment on every
 nontrivial cycle $C$ in $\sigma_1\sigma_2^{-1}$. When all the constraints are
 met, we have
 \[\E[g_\mathrm{avg}^{(1)}(\sigma_1,\sigma_2,\tau_1,\tau_2)]
  =\mu_{2,2}^{c_1(\sigma_1\sigma_2^{-1})}\prod_{\ell=2}^{n}(\mu_{1,1}^{2\ell}/2)^{c_\ell(\sigma_1\sigma_2^{-1})}.\]
 Hence, we have $\E[Z^2_{\B,2}]
 =\sum_{\sigma_1,\sigma_2}\mu_{2,2}^{c_1{(\sigma_1\sigma_2^{-1})}}\prod_{\ell=2}^{n}\mu_{1,1}^{2\ell
 c_\ell(\sigma_1\sigma_2^{-1})} =\Psi_n(\mu_{2,2},\mu_{1,1}^2,1)$.

 \subsection{Case $\mu_{2,0}\neq0$}
 In the case $\mu_{2,0}\neq 0$, the main difference arises from the fixed points of $\sigma_1\sigma_2^{-1}$.
 Recall that
 \[
  g_{\mathrm{avg}}(\sigma_1,\sigma_2)
  =\prod_{i=1}^n \theta_{i,\sigma_1(i)}\theta_{i,\sigma_2(i)}.
 \]
 By independence of the matrix entries, the mixed expectation
 \[
  \E\left[
   g_{\mathrm{avg}}(\sigma_1,\sigma_2)\,
   \overline{g_{\mathrm{avg}}(\tau_1,\tau_2)}
   \right]
 \]
 factorizes over $i\in[n]$. As in the previous case, cancellation occurs on all
 nontrivial cycles of $\sigma_1\sigma_2^{-1}$. However, when $\mu_{2,0}\neq 0$,
 mismatches are allowed at fixed points of $\sigma_1\sigma_2^{-1}$.

 More precisely, the expectation is nonzero if and only if
 \[
  h(\sigma_1,\sigma_2)=h(\tau_1,\tau_2)
  \quad\text{on all nontrivial cycles of }\sigma_1\sigma_2^{-1},
 \]
 while on the fixed-point set
 \[
  \mathrm{Fix}(\sigma_1\sigma_2^{-1})\defeq\{i\in[n]:\sigma_1\sigma_2^{-1}(i)=i\}
 \]
 the values $(\tau_1(i),\tau_2(i))$ may differ from
 $(\sigma_1(i),\sigma_2(i))$. If, for some
 $i\in\mathrm{Fix}(\sigma_1\sigma_2^{-1})$,
 \[
  \sigma_1(i)=\sigma_2(i)=j,
  \qquad
  \tau_1(i)=\tau_2(i)=j'\neq j,
 \]
 then the corresponding factor contributes \(
 \E\Big[\theta_{i,j}^2\overline{\theta_{i,j'}^2}\Big] =|\mu_{2,0}|^2\neq 0. \)

 Fix $\sigma_1,\sigma_2\in \mathcal{S}_n$ and we analyze when the mixed
 expectation does not vanish. For the nontrivial cycles of
 $\sigma_1\sigma_2^{-1}$, the contribution is the same as in the previous case
 and equals \( \prod_{\ell=2}^n \mu_{1,1}^{2\ell
 c_\ell(\sigma_1\sigma_2^{-1})}. \) On the fixed-point set
 $\mathrm{Fix}(\sigma_1\sigma_2^{-1})$, admissible pairs $(\tau_1,\tau_2)$ are
 obtained by permuting the fixed points and independently choosing
 orientations, yielding
 \[
  \big(c_1(\sigma_1\sigma_2^{-1})\big)!\cdot 2^{c(\sigma_1\sigma_2^{-1})}
 \]
 such pairs. Summing over all these $(\tau_1,\tau_2)$ gives
 \begin{align*}
   & \sum_{\tau_1,\tau_2}
  \E\left[
   g_{\mathrm{avg}}(\sigma_1,\sigma_2)\,
   \overline{g_{\mathrm{avg}}(\tau_1,\tau_2)}
  \right]                 \\
   & \qquad =
  2^{c(\sigma_1\sigma_2^{-1})}
  \sum_{\rho\in \mathcal{S}_{c_1(\sigma_1\sigma_2^{-1})}}
  \mu_{2,2}^{c_1(\rho)}
  \prod_{\ell=2}^{c_1(\sigma_1\sigma_2^{-1})} |\mu_{2,0}|^{2\ell c_\ell(\rho)}
  \prod_{\ell=2}^n \mu_{1,1}^{2\ell c_\ell(\sigma_1\sigma_2^{-1})},
 \end{align*}
 where $\rho$ records the cycle structure induced on the fixed points.
 Hence, we have
 \begin{align*}
  \E[Z^2] & =\sum_{\sigma_1,\sigma_2,\tau_1,\tau_2}
  \E\Big[g_\mathrm{avg}(\sigma_1,\sigma_2)\overline{ g_\mathrm{avg}(\tau_1,\tau_2)}\Big] \\
          & =\sum_{\sigma_1,\sigma_2} 2^{c(\sigma_1\sigma_2^{-1})}
  \Big(\sum_{\rho\in  \mathcal{S}_{c_{1}(\sigma_1\sigma_2^{-1})}}\mu_{2,2}^{c_1(\rho)}
  \prod_{\ell=2}^{c_1(\rho)}|\mu_{2,0}|^{2\ell c_\ell(\rho)}\Big)
  \prod_{\ell=2}^{n}\mu_{1,1}^{2\ell c_\ell(\sigma_1\sigma_2^{-1})}                      \\
          & =n!\sum_{\sigma\in \mathcal{S}_n}
  \Big(\sum_{\rho\in  \mathcal{S}_{c_{1}(\sigma)}}\mu_{2,2}^{c_1(\rho)}\prod_{\ell=2}^{c_1(\sigma)}|\mu_{2,0}|^{2\ell c_\ell(\rho)}\Big)
  \prod_{\ell=2}^{n}\mu_{1,1}^{2\ell c_\ell(\sigma)}2^{c_\ell(\sigma)}.
 \end{align*}

 It remains to simplify this expression. View $\rho\in
 \mathcal{S}_{c_1(\sigma)}$ as acting on the fixed points of $\sigma$, thereby
 producing a permutation $\tau\in \mathcal{S}_n$. Under this composition, each
 nontrivial cycle of $\tau$ arises in exactly one of two ways:
 \begin{itemize}
  \item from a nontrivial cycle of $\sigma$, contributing a factor
  \(2(\mu_{1,1}^2)^\ell\);
  \item from a nontrivial cycle of $\rho$ acting on fixed points, contributing
  a factor \(|\mu_{2,0}|^{2\ell}\).
 \end{itemize}
 Moreover, for each resulting $\tau$ there are exactly $2^{c(\tau)}$ pairs
 $(\sigma,\rho)$ producing it.

 Consequently, the sum reorganizes into a single cycle-index expression:
 \[\E[Z^2]=n!\sum_{\tau\in \mathcal{S}_n} \mu_{2,2}^{c_1({\tau})}\prod_{\ell= 2}^n(2\mu_{1,1}^{2\ell}+|\mu_{2,0}|^{2\ell})^{c_\ell(\tau)}=\psi_n(\mu_{2,2}, \mu_{1,1}^2, 2, |\mu_{2,0}|^2).\]

 A similar argument holds for $\E \big[Z_{\B,2}^2\big]$, so we have $\E
 \big[Z_{\B,2}^2\big]=\psi_n (\mu_{2,2}, \mu_{1,1}^2, 1, |\mu_{2,0}|^2)$.


 \section{Proof of Corollary \ref{cor:double:cover:complex}}\label{appx:double:cover:complex:asym}
 The constant $C$ in Corollary \ref{cor:double:cover:complex} is given by
 \[
  C=
  -\frac{\mu_{2,2}-2\mu_{1,1}^2-|\mu_{2,0}|^2}{\mu_{1,1}^2}
  -\frac{|\mu_{2,0}|^2}{\mu_{1,1}^2-|\mu_{2,0}|^2}.
 \]
 Note that $\mu_{1,1}^2-|\mu_{2,0}|^2>0$ always hold by Assumption
 \ref{ass:iid:matrices}, whereas $\mu_{2,2}-2\mu_{1,1}^2-|\mu_{2,0}|^2$ does
 not have a definite sign. The asymptotics of
 \[\sqrt{\frac{\E \big[Z^2(\sfnDE(\btheta))\big]}{\E \big[Z_{\B,2}^2(\sfnDE(\btheta))\big]}}
  =\sqrt{\frac{\psi_n (\mu_{2,2}, \mu_{1,1}^2, 2, |\mu_{2,0}|^2)}{\psi_n (\mu_{2,2}, \mu_{1,1}^2, 1, |\mu_{2,0}|^2)}}\]
 is the consequence of the following lemma.
 \begin{lem}
  For $a,b,c\in \R_{\geq 0}, b > c \geq0$, it holds that
  \begin{align*}
   \psi_n(a,b,2,c) & \sim (n!)^2 \Big(b^n\frac{1}{1-c/b}e^{(a-2b-c)/b}\Big [(n+1)-\frac{a-2b-c}{b}-\frac{c}{c-b} \Big]+\mathcal O(c^n)\Big), \\
   \psi_n(a,b,1,c) & \sim (n!)^2 \Big(b^n\frac{1}{1-c/b}e^{(a-b-c)/b}+\mathcal O(c^n)\Big).
  \end{align*}
 \end{lem}
 \begin{proof}
  Setting $z_1=a, z_k=mb^k+c^k$ for $k\geq 2, m\in \R\setminus\Z_{\leq 0} $, the generating function of $Z_n$ is given by
  \[
   C(x)=\exp\Biggl(\sum_{k\geq 1}\frac{z_kx^k}{k}\Biggr)
   =\exp\bigl((a-mb-c)x-m\ln(1-bx)-\ln(1-cx)\bigr)
   =\frac{e^{(a-mb-c)x}}{(1-bx)^m(1-cx)}.
  \]
  When $b>c$, the singularity of $C(x)$ that determines the asymptotics of
  $[x^n]C(x)$ is $x=1/b$ \cite{flajoletAnalyticCombinatorics2009}.

  Define $f_m(x)\defeq e^{(a-mb-c)x}(1-cx)^{-1}$, which is analytic around
  $x=1/b$ and $f(1/b)\neq 0$. To find an asymptotics with the desired accuracy,
  we consider the Taylor series expansion of $f_m(x)$ around $x=1/b$,
  \[f_m(x)=\sum_{k\geq 0}f_{m,k}(1-bx)^k,\]
  where $f_{m,k}=\frac{(-1/b)^k}{k!}f_m^{(k)}(1/b)$.

  Then we have $C(x)=\sum_{k\geq 0}f_{m,k}(1-bx)^{k-m}$. We compute the
  asymptotics of $[x^n]C(x)$ from $[x^n](1-bx)^{k-m}$:
  \begin{itemize}
   \item For $0\leq k < m$, Thm.~VI.1 in
   \cite{flajoletAnalyticCombinatorics2009} gives more correction terms for the
   asymptotics of $[x^n](1-bx)^{k-m}$. For $m=1$ and $m=2$, the situation is
   much simpler.
   \item For $k\geq m$, $(1-bx)^{k-m}$ is analytic at $x=1/b$, then
   $[x^n](f_{m,k}(1-bx)^{k-m})=\mathcal O(c^n)$. (See
   \cite[Thm.~IV.9\&10]{flajoletAnalyticCombinatorics2009})
  \end{itemize}
  For $m=1$ and 2, we have
  \[\frac{1}{1-bx}=\sum_{n\ge 0}(bx)^n,\frac{1}{(1-bx)^2}=\sum_{n\geq 0}(n+1)(bx)^n \text{ when }|bx|<1.\]
  Hence, we have $[x^n]\frac{1}{1-bx}=b^n$ and
  $[x^n]\frac{1}{(1-bx)^2}=b^n(n+1)$. Putting everything together, we have
  \begin{align*}
   \psi_n(a,b,1,c)
    & =(n!)^2[x^n]\frac{e^{(a-b-c)x}}{(1-bx)(1-cx)}                                                                                 \\
    & \sim(n!)^2(f_{1,0}b^n+\mathcal O(c^n))                                                                                        \\
    & = (n!)^2 \Big(b^n\frac{1}{1-c/b}e^{(a-b-c)/b}+\mathcal O(c^n)\Big),                                                           \\
   \psi_n(a,b,2,c)
    & =(n!)^2[x^n]\frac{e^{(a-2b-c)x}}{(1-bx)^2(1-cx)}                                                                              \\
    & \sim(n!)^2(f_{2,0}b^n(n+1)+f_{2,1}b^n+\mathcal O(c^n))                                                                        \\
    & = (n!)^2 \Big(b^n\frac{1}{1-c/b}e^{(a-2b-c)/b}\Big [(n+1)-\frac{a-2b-c}{b}-\frac{c}{c-b} \Big]+\mathcal O(c^n)\Big). \qedhere
  \end{align*}
 \end{proof}

 \section{Complex-valued Matrix with Zero-mean and Constant-argument Entries}
 \label{appx:double:cover:complex:equal}

 By the Cauchy-Schwarz inequality, for any probability distribution $\mathcal D$
 over $\C$, we have $\mu_{1,1}\geq |\mu_{2,0}|$ when both sides are finite. In
 the case $\mu_{1,1}=|\mu_{2,0}|>0$, the distribution $\mathcal D$ is supported
 over a line through the origin in the complex plane, i.e., $e^{\iota \phi}\R$
 for some constant $\phi\in(-\pi,\pi)$. Theorem~\ref{thm:double:cover:complex}
 still holds, but the asymptotic result is different.

 If we replace the condition $\mu_{1,1}>|\mu_{2,0}|\geq0$ in
 Assumption~\ref{ass:iid:matrices} by $\mu_{1,1}=|\mu_{2,0}|>0$, we obtain the
 following results:
 \begin{align*}
  \E \big[Z^2(\sfnDE(\btheta))\big]        & = \Psi_n (\mu_{2,2}, \mu_{1,1}^2, 3), \\
  \E \big[Z_{\B,2}^2(\sfnDE(\btheta))\big] & = \Psi_n (\mu_{2,2}, \mu_{1,1}^2, 2).
 \end{align*}
 Consequently, we have
 \[
  \sqrt{\frac{\E \big[Z^2(\sfnDE(\btheta))\big]}{\E \big[Z_{\B,2}^2(\sfnDE(\btheta))\big]}}\sim
  \sqrt{\frac{1}{e}\frac{(n+2)(n+1)/2-(a-3b)(n+1)/b+(a-3b)^2/(2b^2)}{(n+1)-(a-2b)/b}},
 \]
 where $a=\mu_{2,2}$ and $b=\mu_{1,1}^2$. The asymptotics is obtained by the
 same method as in the proof of Corollary~\ref{cor:double:cover:complex}.

\else
\fi
\end{document}